\documentclass[journal]{IEEEtran}
\IEEEoverridecommandlockouts
\def\BibTeX{{\rm B\kern-.05em{\sc i\kern-.025em b}\kern-.08em
		T\kern-.1667em\lower.7ex\hbox{E}\kern-.125emX}}
\usepackage[normalem]{ulem}
\usepackage[switch]{lineno}
\usepackage{colortbl}
\usepackage[sort]{cite}
\usepackage[T1]{fontenc}
\usepackage{listings}
\usepackage{courier}
\usepackage{graphicx}
\usepackage{xcolor}
\usepackage{url}
\usepackage{complexity}
\usepackage{amsmath}
\usepackage{amsthm}
\usepackage{amsfonts}
\usepackage{mathtools}
\usepackage[linesnumbered,ruled,vlined]{algorithm2e}
\usepackage{paralist}
\usepackage{enumitem}
\usepackage{pgfplots}
\usepackage[citecolor=blue,bookmarks=false,hypertexnames=true]{hyperref}
\setitemize{label={--},itemsep=.5\parskip,partopsep=0pt,leftmargin=2em}
\usepackage{tikz}
\usetikzlibrary{automata,positioning,calc,arrows.meta}
\tikzset{
        accepting/.style={
            double distance=1mm,
            outer sep=1pt
        }
    }
\usepackage{adjustbox}
\usepackage[T1]{fontenc}
\tikzset{initial text={}}
\usetikzlibrary{decorations.markings}
\usetikzlibrary{calc,positioning,shapes,arrows,3d,fadings,decorations.text, intersections}
\tikzset{every picture/.style={>=angle 60}}
\tikzstyle{ran}=[rounded corners,thick,draw,minimum size=1.4em,inner sep=.5ex]
\tikzstyle{tran}=[thick,draw,->]
\newtheorem{theorem}{Theorem}

\newtheorem{proposition}[theorem]{Proposition}

\newtheorem{example}[theorem]{Example}

\theoremstyle{definition}
\newtheorem{definition}[theorem]{Definition}
\theoremstyle{remark}

\theoremstyle{plain}
\pgfplotsset{compat=1.18} 
\begin{document}
\title{Probabilistic Bisimulation for Parameterized Anonymity and Uniformity Verification}
\author{
    Chih-Duo Hong, Anthony W.~Lin, Philipp Rümmer, Rupak Majumdar
    \thanks{
    Chih-Duo Hong is with the National Chengchi University, Taipei, Taiwan.
    Anthony W.~Lin is with the University of Kaiserslautern-Landau, Kaiserslautern, Germany, and the Max Planck Institute for Software Systems, Kaiserslautern, Germany.
    Philipp Rümmer is with the University of Regensburg, Regensburg, Germany, and Uppsala University, Sweden.
    Rupak Majumdar is with the Max Planck Institute for Software Systems, Kaiserslautern, Germany.
    (Corresponding author: Chih-Duo Hong)
    }
    \thanks{Chih-Duo Hong is supported by the National Science and Technology Council, Taiwan, under grant number NSTC112-2222-E-004-001-MY3. Anthony Lin is supported by the European Research Council under the European Union's Horizon 2020 research and innovation programme under number 101089343. Philipp R\"ummer is supported by the Swedish Research Council through grant 2021-06327.}
}
\maketitle
\begin{abstract}
	Bisimulation is crucial for verifying process equivalence in probabilistic systems. This paper presents a novel logical framework for analyzing bisimulation in probabilistic parameterized systems, namely, infinite families of finite-state probabilistic systems.
	Our framework is built upon the first-order theory of regular structures, which provides a decidable logic for reasoning about these systems.
	We show that essential properties like anonymity and uniformity can be encoded and verified within this framework in a manner aligning with the principles of deductive software verification, where systems, properties, and proofs are expressed in a unified decidable logic.
    By integrating language inference techniques, we achieve full automation in synthesizing candidate bisimulation proofs for anonymity and uniformity.
    We demonstrate the efficacy of our approach by addressing several challenging examples, including cryptographic protocols and randomized algorithms that were previously beyond the reach of fully automated methods.
\end{abstract}
\newcommand{\OMIT}[1]{}
\newcommand{\NEW}[1]{#1}
\newcommand{\DEL}[1]{\sout{#1}}
\newcommand{\constN}{\ensuremath{\mathsf{n}}}
\newcommand{\rev}[1]{\ensuremath{#1^{-1}}}
\newcommand{\citep}[1]{\cite{#1}}
\newcommand{\citealt}[1]{\cite{#1}}
\newcommand{\citeeg}[1]{\cite{#1}}
\newcommand{\citecf}[1]{\cite{#1}}
\newcommand{\size}[1]{\ensuremath{\lvert #1 \rvert}}
\newcommand{\op}[1]{\operatorname{#1}}
\newcommand{\const}[1]{{\sf #1}}
\newcommand{\iset}[1]{\langle #1 \rangle}
\newcommand{\pred}[1]{\mathbb{#1}}
\newcommand{\NatInt}[2]{\ensuremath{\{#1,\ldots,#2\}}}
\renewcommand{\mod}[2]{\ensuremath{#1~{\rm mod}~#2}}
\newcommand{\floor}[1]{\lfloor #1 \rfloor}
\newcommand{\ceil}[1]{\lceil #1 \rceil}
\newcommand{\lowerbnd}[1]{\lfloor{#1}\rfloor}
\newcommand{\upperbnd}[1]{\lceil{#1}\rceil}
\newcommand{\norm}[1]{\lVert#1\rVert}
\newcommand{\nonneg}[1]{#1 \ge 1}
\newcommand{\voc}{V}
\renewcommand{\AP}{\mathsf{AP}}
\newcommand{\ACT}{\mathsf{ACT}}
\newcommand{\nondet}[2]{\mathsf{ndet}(#1,#2)}
\newcommand{\regularWTS}{\iset{S, P, \{r_a \}_{a\in\ACT}}}
\newcommand{\WTS}{\iset{S, P, A, \delta, +}}
\newcommand{\WTSS}{\iset{S, P, \delta}}
\newcommand{\LTS}{\iset{S, A, \delta}}
\newcommand{\TranSystemTriple}{\voc, \phi_{init}, \phi_{tran}}
\newcommand{\TranSystem}{\ensuremath{(\TranSystemTriple)}}
\newcommand{\TranSystemSafe}{\ensuremath{(\TranSystemTriple, \phi_{bad})}}
\newcommand{\TranSystemFair}{\ensuremath{(\TranSystemTriple, \phi_{fair})}}
\newcommand{\UNIV}{\ensuremath{\mathrm{FO}_\mathsf{reg}}}
\newcommand{\args}[1]{\overline{#1}}
\newcommand{\ialphabet}{\ensuremath{\Sigma}}
\newcommand{\struct}{\mathfrak{S}}
\newcommand{\sem}[1]{[\![#1 ]\!]}
\renewcommand{\FO}[1]{\mathrm{FO}(#1)}
\newcommand{\MSO}[1]{\mathrm{MSO}(#1)}
\newcommand{\WMSO}[1]{\mathrm{WMSO}(#1)}
\newcommand{\zero}{\mathtt{0}}
\newcommand{\one}{\mathtt{1}}
\newcommand{\true}{\mathrm{true}}
\newcommand{\false}{\mathrm{false}}
\newcommand{\subst}[3]{#1\, [#2 / #3]}
\newcommand{\eqlen}{\ensuremath{eqL}}
\newcommand{\univ}{\mathfrak{U}}
\newcommand{\univDesc}{\langle \ialphabet^*\!,\: \preceq,\: {\eqlen,}\: \{\prec_a\}_{a \in \ialphabet}\rangle}
\renewcommand{\lang}[1]{\ensuremath{\mathcal{L}(#1)}}
\newcommand{\rel}[1]{\sem{#1}}
\newcommand{\allows}{\,\rhd\,}
\newcommand{\X}{\mathsf{X}}
\newcommand{\Y}{\mathsf{Y}}
\newcommand{\Z}{\mathsf{Z}}
\newcommand{\abs}[1]{#1}
\makeatletter
\DeclareRobustCommand{\circle}{\mathord{\mathpalette\is@circle\relax}}
\newcommand\is@circle[2]{
    \begingroup
    \sbox\z@{\raisebox{\depth}{$\m@th#1\bigcirc$}}
    \sbox\tw@{$#1\square$}
    \resizebox{!}{\ht\tw@}{\usebox{\z@}}
    \endgroup
}
\makeatother
\newcommand{\Mydiamond}[1]{\ensuremath{\langle #1 \rangle}}
\renewcommand{\G}{\square \,}
\newcommand{\F}{\Diamond \,}
\newcommand{\GF}{\square \Diamond \,}
\newcommand{\FG}{\Diamond \square \,}
\newcommand{\Until}{\,\mathit{U}\,}
\newcommand{\Next}{\circle\,}
\newcommand{\Nat}{\mathbb{N}}
\newcommand{\Real}{\mathbb{R}}
\newcommand{\Rat}{\mathbb{Q}}
\newcommand{\Int}{{\mathbb{Z}}}
\newcommand{\T}{\mathsf{T}}
\newcommand{\N}{\mathsf{N}}
\newcommand{\blank}{\mathtt{\#}}
\newcommand{\structA}{\mathfrak{A}}
\newcommand{\structB}{\mathfrak{B}}
\newcommand{\structW}{\mathfrak{W}}
\newcommand{\structT}{\mathfrak{T}}
\newcommand{\structM}{\mathfrak{M}}
\newcommand{\structR}{\mathfrak{R}}
\newcommand{\structH}{\mathfrak{H}}
\newcommand{\bbA}{\mathbb{A}}
\newcommand{\bbE}{\mathbb{E}}
\newcommand{\bbD}{\mathbb{D}}
\newcommand{\bbP}{\mathbb{P}}
\newcommand{\bbQ}{\mathbb{Q}}
\newcommand{\bbI}{\mathbb{I}}
\newcommand{\bbR}{\mathbb{R}}
\newcommand{\bbS}{\mathbb{S}}
\newcommand{\bbZ}{\mathbb{Z}}
\newcommand{\va}{\vec{a}}
\newcommand{\vb}{\vec{b}}
\newcommand{\vc}{\vec{c}}
\newcommand{\vf}{\vec{f}}
\newcommand{\vu}{\vec{u}}
\newcommand{\vv}{\vec{v}}
\newcommand{\vw}{\vec{w}}
\newcommand{\vx}{\vec{x}}
\newcommand{\vy}{\vec{y}}
\newcommand{\vz}{\vec{z}}
\newcommand{\cA}{\mathcal{A}}
\newcommand{\cB}{\mathcal{B}}
\newcommand{\cC}{\mathcal{C}}
\newcommand{\cD}{\mathcal{D}}
\newcommand{\cI}{\mathcal{I}}
\newcommand{\cJ}{\mathcal{J}}
\newcommand{\cE}{\mathcal{E}}
\newcommand{\cF}{\mathcal{F}}
\newcommand{\cG}{\mathcal{G}}
\newcommand{\cV}{\mathcal{V}}
\newcommand{\cM}{\mathcal{M}}
\newcommand{\cQ}{\mathcal{Q}}
\newcommand{\cR}{\mathcal{R}}
\newcommand{\cT}{\mathcal{T}}
\renewcommand{\cP}{\mathcal{P}}
\renewcommand{\cL}{\mathcal{L}}
\renewcommand{\cS}{\mathcal{S}}
\newcommand{\mT}{\mathfrak{T}}
\newcommand{\mS}{\mathfrak{S}}
\newcommand{\mU}{\mathfrak{U}}
\newcommand{\langA}{A}
\newcommand{\langB}{B}
\newcommand{\langT}{T}
\newcommand{\langI}{I}
\newcommand{\langR}{R}
\newcommand{\langE}{Z}
\newcommand{\langL}{L}
\newcommand{\concat}{\cdot}
\newcommand{\+}{\ensuremath{\!\cdot\!}}
\newcommand{\ModelRun}{\ensuremath{\pi}}
\newcommand{\empseq}{\ensuremath{\varepsilon}}
\newcommand{\tran}[1]{\ensuremath{\stackrel{#1}{\longrightarrow}}}
\newcommand{\MEM}{\mathit{Mem}}
\newcommand{\EQ}{\mathit{Equ}}
\newcommand{\Fp}{F_Y}
\newcommand{\myvec}[1]{
\begin{bmatrix}
	#1
\end{bmatrix}
}
\newcommand{\toss}{{\sf toss}}
\newcommand{\tail}{{\sf tail}}
\newcommand{\head}{{\sf head}}
\newcommand{\reset}{{\sf \bot}}
\newcommand{\libalf}[0]{LibAlf}
\newcommand{\eval}[1]{\llbracket#1\rrbracket}
\def\qed {{
        \parfillskip=0pt
        \widowpenalty=10000
        \displaywidowpenalty=10000
        \finalhyphendemerits=0
        \leavevmode
        \unskip
        \nobreak
        \hfil
        \penalty50
        \hskip.2em
        \null
        \hfill
        \qedsymbol\kern-.6pt
        \par}}
\def\qedsymbol{$\square$}

\section{Introduction}
\label{sec:intro}
Formal verification techniques are essential for developing reliable and sustainable software systems.
Among these techniques, bisimulation equivalence provides a powerful means of establishing that two systems exhibit indistinguishable observable behaviors and meet the same requirements specified in expressive modal logic
\cite{milner1980calculus,milner1989communication}.
In software model checking, bisimulation enables proving a system's correctness by showing that the system mirrors the specification step by step \cite{holzmann2004model,banerjee2016deriving}.
In applications involving secrecy and privacy, bisimulation can establish non-leakage of sensitive information
by showing that a potential attacker cannot distinguish between the real protocol's behavior and the idealized confidential behavior \citep{gancher2023core,ramanathan2004probabilistic,castiglioni2020logical}.
Other applications of bisimulation equivalence include information flow analysis \citep{noroozi2019bisimulation,dong2021refinement,salehi2022automated},
knowledge reasoning \citep{aagotnes2023quantifying,tran2014bisimulation,fervari2022bisimulations},
runtime verification \citeeg{scott2023trustworthy},
and computational optimization problems of finite-state automata \citep{bonchi2013checking,kuperberg2021coinductive,jacobs2023fast}
and probabilistic systems \citep{dehnert2013smt,lin2024bisimulation,drakulic2024bq,sproston2006backward}.

Due to its rich applications, the problem of checking bisimulation equivalence has been extensively studied.
This problem is decidable for both probabilistic and nondeterministic \emph{finite-state} systems
\cite{Baier96,valmari2010simple,CBW12}.
For infinite-state systems, such as communication protocols involving an unbounded number of processes,
the problem is generally undecidable \citep{moller1996infinite}.
Therefore, research on the automatic verification of bisimulation for infinite-state systems has primarily taken two directions.
The first focuses on developing heuristics and approximation techniques for the general undecidable problems \cite{kumar2006finite,hong2019probabilistic,abate2024bisimulation}.
The second seeks to identify subclasses of infinite-state systems where bisimulation equivalence remains decidable.
For example, studies like \cite{srba2004roadmap,S05,forejt2018game,koutavas2024pushdown} probe into bisimulation equivalence for pushdown systems and their variants in both probabilistic and non-probabilistic settings.
As noted by Garavel and Lang \cite{garavel2022equivalence}, most research in infinite-state bisimulation has leaned towards theoretical work over developing practical tools.

In this paper, we introduce a first-order framework for reasoning about bisimulation equivalence in probabilistic infinite-state systems. Drawing from recent advances in deductive verification \citeeg{ahrendt2016deductive,hahnle2019deductive,furia2014loop}, we represent these systems in a decidable theory that enables a uniform formalization of the target system, its correctness properties, and the proofs of these properties.
Our key contribution is leveraging the first-order theory of \emph{regular structures} \cite{blumensath2004finite,colcombet2007transforming,lin2022regular} to develop decidable proof rules for probabilistic bisimulation: \NEW{given a binary relation $R$ encoded in this theory, our rules generate verification conditions determining whether $R$ constitutes a probabilistic bisimulation, which can be automatically checked using theorem provers and constraint solvers.}
We showcase this approach's efficacy by examining two essential properties in cryptographic protocols and randomized algorithms: anonymity and uniformity.
    \emph{Anonymity} protects the identities of participants in a protocol. An anonymous protocol guarantees that external observers cannot identify the individuals involved in the interactions, thereby preserving participant confidentiality \cite{beauxis2009probabilistic,bhargava2005probabilistic,halpern2005anonymity}.
    \emph{Uniformity} focuses on the output distribution of a randomized protocol or algorithm, ensuring that the results are evenly distributed across a specified range \cite{barthe2017proving,barthe2020foundations}.

To illustrate the concepts of anonymity and uniformity, consider the \emph{dining cryptographers protocol}. \cite{chaum1988dining}.
This protocol involves $n\ge 3$ participants on a ring, each holding a secret bit.
Let $x_i$ denote the secret bit held by participant $i$ for $i \in \NatInt{0}{n-1}$.
Define the \emph{parity} of these bits as $f(\args{x}) := x_0 \oplus \cdots \oplus x_{n-1}$, where $\oplus$ stands for xor.
The participants aim to compute $f(\args{x})$ without revealing information about the individual bits.
To achieve this, they execute the computation in two stages. (As the participants are arranged on a ring of length $n$, all index arithmetic in the sequel is performed modulo $n$.)
First, each pair of adjacent participants $i$ and $i+1$ computes a random bit $b_i$ that is visible only to them.
Then, each participant $i$ announces $a_i := x_i \oplus b_i \oplus b_{i-1}$ to all participants,
allowing them to compute $f(\args{x})$ since $f(\args{a}) = f(\args{x})$.
The \emph{anonymity property} of the protocol asserts that
any observing participant, say participant $k$, cannot infer the other participants' secrets from observed information,
including her own secret $x_k$, the random bits $b_k$ and $b_{k-1}$, and the announcements $a_0, \ldots, a_{n-1}$.
To establish this property, we may show that the probability distribution of the announcements $a_0, \ldots, a_{n-1}$
is solely determined by the values of $x_k$, $b_{{k}}$,  $b_{{k}-1}$, and $f(\args{x})$. Consequently, no information beyond these values can be inferred from the announcements. If we model the dining cryptographers protocol as a Markov decision process (MDP),
we can verify the anonymity property by showing that any two initial states of the process, say, one with secrets $\args{x}$ and the other with secrets $\args{y}$, are bisimilar as long as $x_k = y_k$ and $f(\args{x}) = f(\args{y})$. Alternatively, we can prove anonymity by showing that, starting from an initial state with secrets $\args{x}$, the announcements observed by participant $k$ are uniformly distributed over
\[
    \{~\args{a} \in \{0,1\}^n : a_k = x_k \oplus b_k \oplus b_{k-1},~f(\args{a}) = f(\args{x})~\},
\]
which indicates that no information beyond the values of $x_k$, $b_k$, $b_{k-1}$, and $f(x)$ can be inferred from the value of $\args{a}$.
Requirements like this are called \emph{uniformity properties}. We provide more details about these properties in Section~\ref{sec:anonymity-examples}.

\NEW{We show that anonymity and uniformity properties can be model checked in a unified manner for a class of infinite-state systems called parameterized systems. A \emph{parameterized system} $\{P_n\}_{n\in N}$
comprises an infinite family of finite-state systems indexed by a parameter $n \in N$
\cite{abdulla2018model,kourtis2024parameterized,lin2016liveness}.
A parameterized system satisfies a property if every instance $P_n$ of the system satisfies that property.} For example, the dining cryptographers protocol is anonymous if it is anonymous for any $n \ge 3$ cryptographers.
\NEW{In this work, we focus on parameterized systems that can be represented within the first-order theory of regular structures, which is expressive enough to capture a range of examples in the literature.} Specifically, we examine the (parameterized) dining cryptographers protocol \citep{chaum1988dining}, grades protocol \citep{apex}, crowds protocol \citep{reiter1998crowds}, random walks, random sums,
Knuth-Yao's random number generator \cite{knuth1976complexity}, and Bertrand's ballot theorem \cite{feller1991introduction}.
Drawing inspiration from software verification, where safety proofs are typically separated into proof rules and inductive invariant synthesis, our proofs for anonymity and uniformity similarly distinguish between proof rules and bisimulation synthesis. We show that
recent advancements in proof generation and refinement \cite{chen2017learning,markgraf2020parameterized,garg2014ice} can be utilized to search for candidate proofs effectively, making our verification procedure truly ``push-button''.

\emph{Contributions.}
The main contributions of this paper are as follows.
Firstly, we propose to employ bisimulation equivalence as a unified framework for model checking anonymity and uniformity --- two properties that have traditionally been addressed with distinct techniques in infinite-state settings.

\NEW{Secondly, we demonstrate that probabilistic systems satisfying the \emph{minimal deviation assumption} \cite{larsen1991bisimulation} (i.e., all transition probabilities are multiples of some $\varepsilon > 0$) can be faithfully encoded in the first-order theory of regular structures. This theory has been previously applied to reason about \emph{qualitative} liveness of MDPs \cite{lin2016liveness}, where the actual probability values are abstracted away. To our knowledge, this work offers the first quantitative encoding of probabilistic systems in the theory.}

Thirdly, \NEW{assuming minimal deviation},
we encode the verification conditions for \NEW{regular} probabilistic bisimulation relations in the aforementioned theory, yielding an algorithmic approach to checking anonymity and uniformity properties.
Indeed, since this theory is syntactically reducible to the weak monadic second-order logic with one successor (WS1S)
\cite{colcombet2007transforming}, our encodings can be manipulated and analyzed using highly optimized tools like \textsc{Mona} \citep{klarlund2001mona} and \textsc{Gaston} \citep{fiedor2017lazy}.

Thus far, our framework has relied on user-provided proofs for verifying anonymity and uniformity.
Our final contribution is demonstrating how language inference algorithms \cite{angluin1987learning,rivest:inference1993,kearns:introduction1994}
can be leveraged to generalize proofs derived from finite system instances to establish correctness for the entire parameterized system. Thanks to these techniques, we successfully verified challenging examples that were previously beyond the reach of fully automated approaches.

This paper is a significant extension of the conference paper \citep{hong2019probabilistic}, which focused solely on the anonymity verification of parameterized MDPs. This current work expands the previous formalism and results by providing a unified logical framework capable of reasoning about both anonymity and uniformity properties. We also offer new case studies to demonstrate the effectiveness of this extended approach.

\section{Preliminaries}
\subsection{Weighted transition system}
A \emph{weighted transition system (WTS)} is a three-sorted structure
$\struct \coloneqq \langle S,P,A,\delta, + \rangle$, where $S$ is a countable set of configurations, $P$ is a countable set of nonnegative numbers, $A$ is a countable set of actions, $+ : P \times P \to P$ defines the addition operation over $P$, and $\delta: S \times A \times S \to P$ is called a \emph{weighted transition function}.
We assume that the elements in $S,P,A$ are \emph{named} for each element $e$, we can effectively find a constant symbol $c_e$ in $\struct$ that is interpreted to $e$.
For configurations $s,t \in S$, we use $\delta(s,a,t)$ to denote the transition weight from $s$ to $t$ via action $a$. \NEW{When $\delta(s,a,t) \neq 0$, we write $s \rightarrow_{a} t$ to indicate that $s$ can move to a successor $t$ in one step.}
A finite or infinite sequence $\pi \coloneqq s_0 \rightarrow_{a_1} s_1 \rightarrow_{a_2} s_2 \rightarrow_{a_3} \cdots$ is called a \emph{path}.
A configuration $s$ is \emph{reachable} from $s'$ if there is a path from $s'$ to $s$. A WTS is~\emph{bounded branching} if for every action, the system can only reach a bounded number of configurations in one step. Namely, there exists a universal bound $b<\infty$ such that for each $s \in S$ and $a \in A$, $|\{ t \in S : s \to_{a} t \}| \leq b$. Since a WTS can have infinitely many actions, a configuration may still have an unbounded number of successors even though the system is bounded branching.

\NEW{
A WTS $\struct \coloneqq \WTS$ is called a \emph{Markov chain} if there is a constant $q\in P$ such that $\sum_{a \in A} \sum_{t \in S} \delta(s,a,t) = q$ for each $s \in S$.
$\struct$ is a \emph{Markov decision process (MDP)} if there is a constant $q\in P$ such that $\sum_{t \in S} \delta(s,a,t) \in \{0, q\}$ holds for each $s \in S$, $a \in A$.}
\NEW{
In both cases, we interpret $\delta(s,a,t)$ as the transition probability
$\delta(s,a,t) / q$ of the transition $s \to_{a} t$. A Markov chain emits an action after it determines a transition. In contrast, a Markov decision process needs to select an action before making a transition. If the process selects action $a$ in state $s$, it moves to a state $t$ with probability $\delta(s,a,t) / q$.
Note that a Markov chain has no terminal state by definition, whilst a Markov decision process can get stuck
by selecting an action $a$ in a state $s$ such that $\sum_{t \in S} \delta(s,a,t) = 0$.
}

\NEW{
In \cite{larsen1991bisimulation},
Larsen and Skou introduced the \emph{minimal deviation assumption}, restricting a probabilistic system such that all transition probabilities share a common divisor.
This assumption holds for most probabilistic systems we encounter, including probabilistic pushdown automata \cite{etessami2009recursive,forejt2018game}, probabilistic parameterized systems \cite{lin2016liveness}, and all case studies in Section \ref{sec:case-study}.
Notably, a probabilistic system with minimal deviation induces a bounded branching WTS with natural weights.
We will adopt this assumption in the sequel and later discuss how to relax it in Section \ref{sec:extensions}.}
\subsection{Bisimulation equivalence}
\label{sec:bisimulation-PML}
Let $\struct \coloneqq \WTS$ be a WTS.
A \emph{bisimulation} over $\struct$ is an equivalence relation $R \subseteq S \times S$ such that $(s,s')\in R$ implies
\begin{align}\label{defn:pb-proof-rule}
    \forall a\in A.\ \forall E \in S/R.\ \sum_{t \in E} \delta(s, a, t) = \sum_{t' \in E} \delta(s', a, t'),
\end{align}
where $S/R$ denotes the set of equivalence classes induced by $R$. Two configurations $s, s'$ of $\struct$ are \emph{bisimilar} if there exists a bisimulation $R$ over $\struct$ such that $(s, s')\in R$. Intuitively, bisimilar configurations emit the same amount of probability mass to the same equivalence class for any action.
The union of all bisimulations over $\struct$ is itself a bisimulation over $\struct$;
this maximal bisimulation is also called \emph{bisimulation equivalence} or \emph{bisimilarity}
\cite{milner1989communication}.

Bisimilarity can be lifted to relate two systems.
Given two WTSs $\struct \coloneqq \WTS$ and $\struct' \coloneqq \langle S', P, A, \delta', +\rangle$
such that $S \cap S' = \emptyset$, we define the~\emph{disjoint union} of $\struct$ and $\struct'$
as the WTS $\struct\uplus\struct' \coloneqq \iset{ S'', P, A, \delta'', +}$ with $S'' \coloneqq S \uplus S'$ and
\[
    \delta''(s,a,t) \coloneqq \begin{cases}
    \ \delta(s,a,t), & s,t \in S \,;\\
    \ \delta'(s,a,t), & s,t \in S'.
    \end{cases}
\]
A binary relation $R$ over $S \uplus S'$ is called a bisimulation between $\struct$ and $\struct'$ if $R$ is a bisimulation over $\struct\uplus\struct'$.

\NEW{\emph{Probabilistic modal logic (PML)} \cite{larsen1991bisimulation}, originally introduced for probabilistic systems, extends classical propositional logic with formulas of the form $\langle a \rangle_p \phi$, where $p$ is a probability. A state satisfies $\langle a \rangle_p \phi$ if it can move to states satisfying $\phi$ with probability exceeding $p$
through an $a$-labeled transition.
Many results of PML can be naturally generalized to WTSs. In this work, we will exploit the following key property of PML.}
\begin{proposition}[\hspace{1sp}\citealt{bianco1995model,clerc2019expressiveness}]
    \label{thm:PML-finite-WTS}
    \NEW{Two configurations of a WTS are bisimilar if and only if they satisfy the same PML formulas.
    Furthermore, it is decidable to check whether a configuration satisfies a PML formula in a finite WTS.}
\end{proposition}
\OMIT{
\NEW{Indeed, PML can be translated to \emph{probabilistic computation tree logic (PCTL)} \cite{bianco1995model} in linear time.
PCTL model checking runs in polynomial time on finite-state systems, and the algorithm extends to WTSs with a polynomial overhead. Hence, we can model check PML in polynomial time on finite WTSs.}}
\subsection{A first-order framework for regular relations}
An \emph{alphabet} $\Sigma$ is a finite set of letters, and a \emph{word} is a finite sequence of letters.
We use $\Sigma^*$ to denote the set of words over alphabet $\Sigma$. Note that $\Sigma^*$ contains $\varepsilon$, the empty word.
Given $w \in \Sigma^*$, $\size{w}$ is the length of $w$ and $w[i]$ is the $i$th letter of $w$ for $i \in \NatInt{1}{w}$.
We define $\Sigma^n \coloneqq \{(a_1,\ldots,a_n) : a_1,\ldots,a_n \in \Sigma \}$ as the set of $n$-tuples over $\Sigma$,
and denote $\Sigma_{\blank}$ by $\Sigma \uplus \{\blank\}$,
where $\blank \notin \Sigma$ is the \emph{blank symbol}.
The \emph{convolution} $w_1\otimes \cdots \otimes w_n$
of $n$ words $w_1,\ldots,w_n \in\Sigma^*$ is the word $w \in (\Sigma_{\blank}^n)^*$ satisfying
(i) $\size{w} = \max \{\size{w_1}, \dots, \size{w_n} \}$, and
(ii) $w[i] = (a_1,\ldots,a_n)$ for $i \in \NatInt{1}{\size{w}}$, where
\[
a_k \coloneqq \left\{ \begin{array}{cc}
w_k[i] & \quad \text{$\size{w_k} \geq i$,} \\
\blank   & \quad \text{otherwise.}
\end{array}
\right.
\]
In other words, $w$ is the shortest word
obtained by juxtaposing $w_1, \ldots, w_n$ and padding
the shorter words with the blank symbol $\blank$. For example,
$abb \otimes aa = (a,a) (b,a) (b,\blank) \in (\Sigma^2_\blank)^*$
when $\{a,b\} \subseteq \Sigma$.
Given an $n$-ary relation $R \subseteq (\Sigma^*)^n$ over the words $\Sigma^*$,
we define the \emph{language representation} of $R$ as
\begin{align}
    \lang{R} \coloneqq \{ w_1 \otimes \cdots \otimes w_n : (w_1,\ldots,w_n) \in R\}.
\end{align}
We say that an $n$-ary relation $R \subseteq (\Sigma^*)^n$ is a~\emph{regular relation} if
$\lang{R}\subseteq (\Sigma_{\blank}^n)^*$ is a regular language.
We define a structure $\univ \coloneqq \univDesc$, where $\preceq$ is the prefix-of relation,
$\eqlen$ is the equal-length relation, and $\prec_a$
is the $a$-successor relation. More precisely,
$w \preceq w'$ holds iff there exists $w''\in\Sigma^*$ such that $w\cdot w'' = w'$,
where $\cdot$ denotes word concatenation.
Furthermore, $\eqlen(w,w')$ holds iff $|w| = |w'|$, and
$w \prec_a w'$ holds iff $w\cdot a = w'$.
Given a formula $\phi(x_1,\dots,x_n)$ defined over $\univ$, we use $\sem{\phi}$ to denote the set
$\{ (w_1,\ldots,w_n) : \phi(w_1,\ldots,w_n) \in \FO{\univ} \} \subseteq (\Sigma^*)^n$,
where $\FO{\univ}$ denotes the first-order theory of $\univ$.
The following result establishes a connection between logic and automata.
\begin{proposition}[\hspace{1sp}\citealt{blumensath2004finite,colcombet2007transforming}]
    \label{thm:logic-characterisation-of-regularity}
    For any formula $\phi(x_1,\dots,x_n)$ defined over $\univ$, $\sem{\phi} \subseteq (\Sigma^*)^n$ is a regular relation,
    and we can compute from $\phi$ a finite automaton $\mathcal{A}$ recognizing $\lang{\sem{\phi}}$.
    Conversely, for a regular relation $R \subseteq (\Sigma^*)^n$, suppose that $\mathcal{A}$ is a finite automaton recognizing $\lang{R}$.
    Then we can compute from $\mathcal{A}$ a formula $\phi(x_1,\dots,x_n)$ over $\univ$ such that $\sem{\phi} = R$.
\end{proposition}
An immediate consequence of Proposition~\ref{thm:logic-characterisation-of-regularity}
is that $\FO{\univ}$ is decidable, as can be shown using automata-theoretic arguments \citealt{hong2022symbolic}.
The theory remains decidable even if we extend the structure $\univ$ with arbitrary regular relations.
In the sequel, we shall regard relations definable in $\univ$ (such as membership in a regular language) as regular relations, and freely use regular relations as syntactic sugar when we are defining a formula over $\univ$.
For a fixed alphabet $\Sigma$, we shall use {\UNIV} to denote the decidable first-order theory of regular relations over $\Sigma$.

A structure $\struct$ is \emph{regular} if its relational variant is isomorphic to a relational structure $\structT$ such that
(i) the universes of $\structT$ are regular languages over a finite alphabet $\Sigma$, and
(ii) all relations in $\structT$ are regular.
In such case, we call
$\structT$ a \emph{regular presentation} of~$\struct$.
It follows by Proposition~\ref{thm:logic-characterisation-of-regularity} that
the first-order theory of a regular structure is decidable.

\section{Bisimulation in regular relations}
\label{sec:model}
In this section, we establish how weighted transition systems (WTSs) and bisimulation proof rules can be systematically formulated within our logical framework. We begin by specifying regular WTSs using the first-order theory of regular relations, followed by the definition of verification conditions within the same formalism.
These results provide a unified approach for representing and reasoning about WTS properties and proofs.
\subsection{The non-weighted case}
As a warm-up, we first consider bisimulation relations over non-weighted transition systems.
A \emph{labeled transition system (LTS)} is a structure $\struct \coloneqq \iset{S, A, \delta }$ with a labeled transition relation $\delta \subseteq S \times A \times S$.
For convenience, denote $(s,a,t) \in \delta$ as $s \to_a t$.
Then a binary relation $R \subseteq S \times S$ is a bisimulation on the LTS $\struct$
if for all $a\in A$, $t,t' \in S$, and $(s, s') \in R$:
\vspace{.5em}\\\vspace{.2em}
(i) $s \to_a t$ only if $s' \to_a t''$ and $(t,t'')\in R$ for some $t''$;\\\vspace{.5em}
(ii) $s' \to_a t'$ only if $s \to_a t''$ and $(t'',t')\in R$ for some $t''$.\\\vspace{.5em}
These conditions can be expressed as a first-order formula:
\begin{align*}
 \psi(s, a, s') & \coloneqq (\forall t.\, \delta(s, a, t) \Rightarrow \exists t''.\, \delta(s', a, t'') \wedge R(t, t'')) \\
                & \quad \wedge (\forall t'.\, \delta(s', a, t') \Rightarrow \exists t''.\, \delta(s, a, t'') \wedge R(t'', t')).
\end{align*}
Given a binary relation $R \subseteq S \times S$,
let $\struct_R$ denote the structure obtained by extending $\struct$ with $R$.
Then $R$ is a bisimulation on $\struct$ if and only if
\begin{equation}
    \label{sb-proof-rule}
    \struct_R \models \forall s.\,\forall s'. \left( R(s,s') \Rightarrow \forall a.~\psi(s, a, s')\right).
\end{equation}
It is algorithmic to check (\ref{sb-proof-rule}) when both $\struct$ and $R$ are regular.
Indeed, as bisimilarity is preserved under isomorphism, we can verify the sentence by reasoning about the regular presentation of $\struct_R$ in the decidable theory {\UNIV}.

\NEW{While validating a regular bisimulation for a regular LTS is decidable, checking bisimilarity for a regular LTS is not.
To see this, notice that the configuration graph of a Turing machine (TM) is regular \cite{barany2007automatic}:
it can be modeled by a bounded-branching regular LTS with a dummy action $a$, such that each configuration encodes the current state and tape content, and the initial configuration $s_0$ stores an input $x$.
Observe that $s_0$ is bisimilar to a configuration $s$ with a single outgoing transition $s \to_a s$ if and only if the TM does not halt on $x$. Therefore, the halting problem can be reduced to bisimulation checking in a regular LTS, rendering the latter problem undecidable.}
\subsection{Specifying and validating weighted systems}
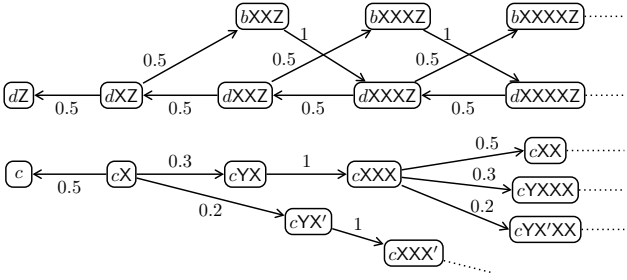
\begin{figure}
\centering
\scalebox{0.7}{
\begin{tikzpicture}[x=2.5cm,y=1.5cm]
\foreach \x/\d/\l in {0/0/d\Z,1/.22/d\X\Z,2/.28/d\X\X\Z,3/.2/d\X\X\X\Z,4/0/d\X\X\X\X\Z}
   \node (D\x)   at (\x-\d,-1)   [ran] {$\l$};
\foreach \x/\xx in {1/0,2/1,3/2,4/3}
   \draw [tran] (D\x) to  node[below]   {$0.5$} (D\xx);
\foreach \x/\d\l in {2/.15/b\X\X\Z,3/.12/b\X\X\X\Z,4/0/b\X\X\X\X\Z}
   \node (H\x)   at (\x-\d,0)   [ran] {$\l$};
\foreach \x/\xx in {1/2,2/3,3/4}
   \draw [tran,rounded corners] (D\x) -- node[very near start,above=0.7mm] {$0.5$} (H\xx);
\foreach \x/\xx in {2/3,3/4}
   \draw [tran,rounded corners] (H\x) -- node[near start,above=0.5mm,inner sep=0] {$1$} (D\xx);
\draw [thick, dotted] (H4) -- +(.6,0);
\draw [thick, dotted] (D4) -- +(.6,0);
\foreach \x/\d/\l in {0/0/c,1/.22/c\X,2/.28/c\Y\X,3/.3/c\X\X\X}
   \node (A\x)   at (\x-\d,-2)   [ran] {$\l$};
   \node (Ba)   at (2.2,-2.6)   [ran] {$c\Y\X'$};
   \node (Bb)   at (3,-3)   [ran] {$c\X\X\X'$};
   \node (Aa)   at (4,-1.7)   [ran] {$c\X\X$};
   \node (Ab)   at (4,-2.2)   [ran] {$c\Y\X\X\X$};
   \node (Ac)   at (4,-2.7)   [ran] {$c\Y\X'\X\X$};
   \draw [tran,rounded corners] (A1) -- node[midway,below] {$0.5$} (A0);
   \draw [tran,rounded corners] (A1) -- node[midway,above] {$0.3$} (A2);
   \draw [tran,rounded corners] (A1) -- node[midway,below] {$0.2$} (Ba);
   \draw [tran,rounded corners] (A2) -- node[midway,above] {$1$} (A3);
   \draw [tran,rounded corners] (A3) -- node[pos=0.7,above] {$0.5$} (Aa.west);
   \draw [tran,rounded corners] (A3) -- node[near end,above] {$0.3$} (Ab.west);
   \draw [tran,rounded corners] (A3) -- node[pos=0.75,above] {$0.2$} (Ac.west);
   \draw [tran,rounded corners] (Ba) -- node[midway,above] {$1$} (Bb);
\draw [thick, dotted] (Aa) -- +(.6,0);
\draw [thick, dotted] (Ab) -- +(.6,0);
\draw [thick, dotted] (Ac) -- +(.6,0);
\draw [thick, dotted] (Bb) -- +(.6,-0.25);
\end{tikzpicture}}
\caption{{Part of the configuration graph in Example~\ref{ex:pPDA}, adapted from \cite{forejt2018game}.}}
\label{fig:pPDA}
\end{figure}
Weighted transition systems introduce additional complexity compared to their non-weighted counterparts,
requiring careful specification and validation.
Below, we examine how to formally define such systems, and illustrate this process through the regular presentation of a simple probabilistic system.
\OMIT{
While the definition of regular WTSs  allows for an infinite number of transition weights,
a subclass of regular WTSs is of practical interest and worth mentioning.
A WTS is said to satisfy the~\emph{minimal weight assumption} \citecf{larsen1991bisimulation} if the weights of the WTS are multiples of some $\epsilon > 0$.
This assumption is practically sensible since it is satisfied by most WTSs we encounter, e.g., finite WTSs, probabilistic pushdown automata~\cite{EE04}, and most examples from probabilistic parameterized systems~\cite{lin2016liveness, lengal2017fair} including our case studies in Section~\ref{sec:anonymity-examples}.
When a WTS satisfies the minimal weight assumption, we can convert its weights to natural numbers by factoring.
This trick is known in the literature of probabilistic verification (see e.g.~\cite{AHM07}), and is sound for the purpose of model-checking bisimulation.
When a regular WTS satisfies the minimal weight assumption, we can assume without loss of generality that the weights are natural numbers,
as is shown in the following example.
}
\begin{example}
\label{ex:pPDA}
\NEW{
Pushdown automata (PDAs) are finite-state machines equipped with a stack over a finite set of symbols.
Each transition in a PDA can modify the stack by pushing or popping a symbol.
Probabilistic pushdown automata (pPDAs), a.k.a.~recursive Markov chains \cite{etessami2009recursive}, further enrich PDAs with transition probabilities. Consider a pPDA from Example 1 of \cite{forejt2018game}, which has states $Q = \{b, c, d\}$ and stack symbols $\Gamma = \{{\X}, {\X}', {\Y}, {\Z}\}$.
A configuration $qs \in Q\Gamma^*$ is a word that comprises the current state $q\in Q$ and stack content $s \in \Gamma^*$. The transition rules are given by:
\[
\begin{array}{llll}
d{\X} \xrightarrow{0.5} b{\X}{\X} & ~d{\X}  \xrightarrow{0.5} d & ~b{\X} \xrightarrow{1} d{\X}{\X} & ~c{\Y} \xrightarrow{1} c{\X}{\X} \\
c{\X} \xrightarrow{0.3} c{\Y}{\X} & ~c{\X}  \xrightarrow{0.2} c{\Y}{\X}' & ~c{\X} \xrightarrow{0.5} c \\
c{\X}'\xrightarrow{0.4} c{\Y}{\X} & ~c{\X}' \xrightarrow{0.1} c{\Y}{\X}' & ~c{\X}' \xrightarrow{0.5} c
\end{array}
\]
A rule is applicable if the configuration's prefix matches the left-hand side of the rule.

Figure~\ref{fig:pPDA} presents a fragment of the pPDA's configuration graph.
We can model this graph as a WTS in {\UNIV}:}
the configuration set is $Q \Gamma^*$,
the weights are encoded in binary after normalizing the probabilities to natural numbers.
The transition relation is encoded as a disjunctive formula $\delta(s, a, t, p)$, e.g., the disjunct corresponding to the rule $\,d{\X} \xrightarrow{0.5} b{\X}{\X}\,$ is
$\exists u \in \Gamma^*.\,(s = d{\X}u \wedge t = b{\X}{\X}u \wedge p = \one\zero\one)$.
The branching of the resulting WTS is bounded by $3$.

\NEW{Note that {\UNIV} is expressive enough to represent any pPDA with rational transition probabilities,
since these probabilities can be encoded as natural weights in its regular presentation, as we have exemplified above.
See \cite{barany2007automatic} for further details on infinite-state systems with regular presentations.}
\end{example}
\NEW{We can effectively check whether a set of {\UNIV} formulas properly specifies a WTS $\struct \coloneqq \WTS$, provided that the branching bound of $\struct$ is known.}
\NEW{Suppose we have {\UNIV} formulas defining $S$, $P$, $A$, and a formula $\phi(s, a, t, p)$ defining the transition function $\delta$.
To check that these formulas indeed specify a WTS with a branching bound $n$, we essentially need to verify the following three conditions:}
\begin{enumerate}
	\item The branching is indeed bounded by $n$. That is, for each $s\in S$ and $a \in A$, there are at most $n$ distinct $t$'s in $S$ such that $\phi (s, a, t, p) \wedge p \neq 0$ is satisfiable.
	\item $\phi$ encodes a function of sort $S \times A \times S \to P$. That is, for each $s,t\in S$ and $a \in A$, there exists precisely one $p\in P$ satisfying $\phi(s,a,t,p)$.
	\item When the WTS is a Markov chain or a Markov decision process, $\phi$ should encode a mapping from $S$ to the set of probability distributions over $S$.
\end{enumerate}
The first two conditions are expressible in {\UNIV} and hence are algorithmic by Proposition~\ref{thm:logic-characterisation-of-regularity}. To verify the third condition, we define an auxiliary formula $\psi(s, a, \args{t})$, which asserts that the successors of $s$ through action $a$ are among the configurations $\args{t} \coloneqq t_1,\ldots, t_n :$
\begin{align*}
    \psi(s, a, \args{t}) \coloneqq \forall t.\,\forall p.\left(
    \phi(s, a, t, p) \wedge p \neq 0 \Rightarrow \bigvee_{i\,} (t = t_i)\right).
\end{align*}
The third condition then amounts to checking that there exists a number $q \in P$ such that
\begin{align*}
    \forall s.\forall a.
    & \left( \exists \args{t}. \exists \args{p}.
    \left(\psi(s, a, \args{t}) \wedge \bigwedge_{i} \phi(s, a, t_i, p_i) \wedge \sum_{\,i} p_i = q \right)\right)
\end{align*}
holds when the WTS is a Markov chain, and
\begin{align*}
    & \forall s.\forall a.\left(\forall t.\forall p.\, (\phi(s, a, t, p) \Leftrightarrow p = 0)\right) \\
    & ~\vee \left( \exists \args{t}.\exists \args{p}.
    \left( \psi(s, a, \args{t}) \wedge \bigwedge_{i} \phi(s, a, t_i, p_i) \right) \wedge \sum_{\,i} p_i = q\right)
\end{align*}
holds when the WTS is a Markov decision process. Again, this check is algorithmic in {\UNIV}.
Hence, it is decidable to check whether a given regular presentation of a WTS is well-defined.
\subsection{Proof rules for bisimulation over weighted systems}
Recall that $\struct_R$ denotes the structure obtained by extending structure $\struct$ with relation $R$.
The following theorem summarizes the main technical result of this work.
\begin{theorem}
    \label{th:verify}
    There is a fixed first-order sentence $\Phi$ such that a given binary relation $R$ is a bisimulation on a WTS $\struct$ if and only if $\struct_R \models \Phi$. Furthermore, checking $\struct_R \models \Phi$ is decidable when both $\struct$ and $R$ are regular.
\end{theorem}
\begin{proof}
Fix a WTS~$\struct \coloneqq \WTS$ with branching bound $n$.
The fact that a binary relation $R$ is an equivalence relation
can be expressed by
	$\Phi_\mathsf{eq} \coloneqq \forall s. \forall s'. \forall s''.~R(s,s)
	\wedge (R(s,s') \Rightarrow R(s',s))
	\wedge (R(s,s') \wedge R(s',s'') \Rightarrow R(s,s'')).$
Now, it suffices to define a sentence $\Phi$ as
\begin{equation}
    \label{def:Phi}
    \Phi_\mathsf{eq} \wedge \forall s. \forall s'.~R(s, s') \Rightarrow
    \forall a.\,(\psi(s, a, s') \vee \lambda(s, a, s')),
\end{equation}
such that $R$ is a bisimulation over $\struct$ if and only if
$\struct_R \models \Phi$.
Here, $\psi(s,a,s')$ is an auxiliary formula asserting that configurations $s$ and $s'$ have no successor through action $a$:
\begin{equation}
    \label{def:psi}
    \psi(s,a,s') \coloneqq \forall t.\,(\delta(s,a,t) = 0 \wedge \delta(s',a,t) = 0).
\end{equation}
Before defining $\lambda(s,a,t)$, we first offer some intuition and auxiliary formulas.
Given configurations $s$ and $t$, the formula $\lambda(s,a,t)$ will first guess a set of $n$ configurations $\args{u}$ containing the successors of $s$ through action $a$, and a set of $n$ configurations $\args{v}$ containing the successors of $t$ through action $a$.
The formula will also guess a labeling $\args{\alpha}$ and $\args{\beta}$ that corresponds to the partitioning of the configurations
$\args{u}$ and $\args{v}$, respectively.
The intuition here is that the labeling ``names'' the partitions:
$\alpha_i = \alpha_j$ (resp.~$\beta_i = \beta_j$)
means that $u_i$ and $u_j$ (resp.~$v_i$ and $v_j$)
are guessed to be in the same partition.
The formula then checks that the guessed partitioning is compatible with the equivalence relation $R$
(i.e.~$u_i$ and $u_j$ have the same label iff $R(u_i, u_j)$ holds),
and that the probability masses of the partitions assigned by configurations $s$ and $t$
satisfy the constraint given at (\ref{defn:pb-proof-rule}).

For the labeling, define an auxiliary formula $\mathsf{succ}(s,a,\args{u})$:
\[
    \left( \bigwedge_{i < j} u_i \neq u_j \right) \wedge
    \left( \forall t.\:\delta(s, a, t) \neq 0 \Rightarrow \bigvee_{\,i} t = u_i \right),
\]
stating that the successors of configuration $w$ on action $a$ are among the $n$
distinct configurations $\args{u}$.
Note that a configuration may have fewer than $n$ successors.
In this case, we can set the rest of the variables to arbitrary distinct configurations.
Given a labeling, we need to check that $R$ is compatible with the guessed partitions,
and that configurations $s$ and $t$ assign the same probability mass to the same partition.
Let $\args{k}$ be a labeling for configurations $\args{s}$.
To check that the partitioning induced by the labeling is compatible with $R$,
we need to express the condition that $k_i = k_j$ if and only if $R(s_i, s_j)$ holds.
This condition can be expressed as a formula
\[
    \mathsf{compat}(\args{s}, \args{k}) \coloneqq
    \bigwedge_{i<j} \left( R(s_i, s_j) \Leftrightarrow k_i = k_j \right).
\]
Now, we can define the formula $\lambda(s, a, s')$ at (\ref{def:Phi}) as
\begin{align}
\label{def:lambda}
		& \exists \args{u}.\,\exists \args{v}.\,\exists \args{\alpha}.\,\exists \args{\beta}.~\mathsf{succ}(s, a, \args{u}) \wedge \mathsf{succ}(s', a, \args{v}) \nonumber\\
        & \qquad \wedge \mathsf{compat}(\args{u}, \args{\alpha}) \wedge \mathsf{compat}(\args{v}, \args{\beta}) \\
        & \qquad \wedge \forall k.
        \left(\sum_{\,i:\:\alpha_i = k} \delta(s, a, u_i) = \sum_{\,i:\:\beta_i = k} \delta(s', a, v_i)\right). \nonumber
\end{align}
Thus, $\struct_R \models \lambda(s,a,s')$ iff
$\sum_{t\in E} \delta(s,a,t) = \sum_{t\in E} \delta(s',a,t)$ holds for any equivalence class $E \in S/R$.
It follows that the sentence $\Phi$ characterizes a bisimulation relation,
and $\struct_R \models \Phi$ if and only if $R$ is a bisimulation over $\struct$.

We proceed to show that $\Phi$ is expressible in ${\UNIV}$. Translating $\Phi_\mathsf{eq}$ and $\psi$ to ${\UNIV}$ is straightforward. To translate $\lambda$, the key step is to express the summation inside (\ref{def:lambda}). For this, we define a formula that performs iterated additions:
\begin{align*}
    \chi(s, a, \args{t}, \args{\alpha}, k, z)
        & \coloneqq \exists \args{p}.~p_1 = 0 \wedge p_{n+1} = z \\
        & \qquad \wedge \bigwedge_{1\le i \le n} \chi'(s, a, t_i, \alpha_i, k, p_i, p_{i+1}),
\end{align*}
where $s \in S$, $\args{t} \in S^{n}$, $k\in P$, $\args{\alpha} \in P^n$, $\args{p} \in P^{n+1}$, and
\begin{align*}
    & \chi'(s, a, t, \kappa, k, x, y) \coloneqq (\kappa \neq k \wedge x = y)\ \vee \\
    & \qquad\qquad\qquad (\kappa = k \wedge \exists p.\:\delta(s, a, t) = p \wedge (x + p = y)).
\end{align*}
The formula $\bigwedge_{1\le i \le n} \chi'(s, a, t_i, \alpha_i, k, p_i, p_{i+1})$ effectively sums up the weights in $\{\,\delta(s, a, t_i) : \alpha_i = k,\, 1\le i \le n\,\}$ and stores the result in $p_{n+1}$.
Thus, given $k\in P$, $\args{u}, \args{v} \in S^{n}$ and $\args{\alpha}, \args{\beta} \in P^n$,
we have
\[
    \sum_{\,i:\:\alpha_i = k} \delta(s, a, u_i) = \sum_{\,i:\:\beta_i = k} \delta(s', a, v_i)
\]
if and only if
\[
    \struct \models \exists z.\ \chi(s, a, \args{u}, \args{\alpha}, k, z)
        \wedge \chi(t, a, \args{v}, \args{\beta}, k, z),
\]
and the definition at (\ref{def:lambda}) can be equivalently written as
\begin{align*}
		& \exists \args{u}.\,\exists \args{v}.\,\exists \args{\alpha}.\,\exists \args{\beta}.~\mathsf{succ}(s, a, \args{u}) \wedge \mathsf{succ}(s',a, \args{v}) \nonumber\\
        & \qquad \wedge \mathsf{compat}(\args{u}, \args{\alpha}) \wedge \mathsf{compat}(\args{v}, \args{\beta}) \\
        & \qquad \wedge \forall k.\,\exists z.\: \chi(s, a, \args{u}, \args{\alpha}, k, z) \wedge \chi(s',a, \args{v}, \args{\beta}, k, z).
\end{align*}
Consequently, the sentence $\Phi$ at (\ref{def:Phi}) is definable in $\FO{\struct_R}$.
By Proposition~\ref{thm:logic-characterisation-of-regularity}, it is decidable to check $\struct_R \models \Phi$ when both $\struct$ and $R$ are regular. This concludes our proof.
\end{proof}
\begin{example}
Consider the regular WTS from Example~\ref{ex:pPDA}.
Note that the configurations $d{\X}{\Z}$ and $c{\X}$ are bisimilar.
This fact can be shown using a bisimulation relation with equivalence classes
$\{d{\X}^k{\Z}\} \cup \{dw : w \in \{{\X}, {\X}'\}^k\}$
and
$\{b{\X}^{k+2}{\Z}\} \cup \{c{\Y} w : w \in\{{\X}, {\X}'\}^{k+1}\}$
for all $k\geq 0$.
This bisimulation relation is definable as the reflexive and symmetric closure of a regular relation $R$, where $(v, u)\in R$ if and only if
\begin{align*}
    & (v \in d{\X}^*{\Z} \wedge u \in c({\X}+{\X}')^* \wedge |v| = |u| + 1)\\
    & \vee (v \in c({\X}+{\X}')^* \wedge u \in c({\X}+{\X}')^* \wedge |v| = |u|)\\
    & \vee (v \in b{\X}^*{\Z} \wedge u \in c{\Y}({\X}+{\X}')^* \wedge |v| = |u| + 1)\\
    & \vee (v \in c{\Y}({\X}+{\X}')^* \wedge u \in c{\Y}({\X}+{\X}')^* \wedge |v| = |u|).
\end{align*}
The formula $\lambda(s,a,s')$ at (\ref{def:lambda}) checks this bisimulation relation for all states.
To see the formula in action, fix two bisimilar configurations $c{\X}$ and $d{\X}{\Z}$.
In the WTS, $c{\X}$ has three successors, $c{\Y}{\X}$, $c{\Y}{\X}'$, and $c$, with probabilities $0.3$, $0.2$, and $0.5$, respectively;
$d{\X}{\Z}$ has two successors, $b{\X}{\X}{\Z}$ and $d{\Z}$, each with probability $0.5$.
These successors form two equivalence classes $\{d{\Z},c\}$ and $\{b{\X}{\X}{\Z},c{\Y}{\X},c{\Y}{\X}'\}$.
To satisfy formula $\lambda$,
let $s=c{\X}$ with successors $u_1=c{\Y}{\X}$, $u_2=c{\Y}{\X}'$, $u_3=c$.
Let $s'=d{\X}{\Z}$ with successors $v_1=b{\X}{\X}{\Z}$, $v_2=d{\Z}$, and $v_3$ being any configuration not equal to $v_1$ and $v_2$.
\NEW{To label these successors,
we can set
$\alpha_3 = \beta_2 = 1$,
$\alpha_1 = \alpha_2 = \beta_1 = 2$, and $\beta_3$ to an arbitrary number other than $1$ and $2$.
Now, for $k \notin \{1,2\}$,
$\sum_{\,i:\:\alpha_i = k} \delta(s, a, u_i) = \sum_{\,i:\:\beta_i = k} \delta(s', a, v_i)$ yields $0=0$.
For $k=1$, it yields $\delta(s,a,u_3)=\delta(s',a,v_2)=0.5$. For $k=2$, it yields $\delta(s,a,u_1)+\delta(s,a,u_2)=0.3+0.2=\delta(s',a,v_1)=0.5$.
Thus, $\lambda(s,a,s')$ confirms that $c\X$ and $d\X\Z$ are bisimilar.}
\end{example}
Theorem~\ref{th:verify} leads to the following decidability result.
\begin{theorem}
   \label{thm:synthesis}
    Given a regular WTS $\struct$ and a regular relation $E \subseteq S \times S$,
    there exists a procedure that finds either a non-bisimilar pair $(u, v)\in E$,
    or a regular bisimulation relation $R$ over $\struct$ such that $E \subseteq R$.
    Furthermore, the procedure terminates if and only if $E$ contains a non-bisimilar pair
    or $E$ is a subset of a regular bisimulation relation.
\end{theorem}
\begin{proof}
    It suffices to give two semi-procedures, one checking the
    existence of $R$ and the other finding a non-bisimilar pair $(v,w)$ in $E$.
    By Theorem \ref{th:verify}, we can enumerate all plausible regular
    relations $R$ and check that $R$ is a bisimulation over $\struct$.
    The inclusion of $E$ in $R$ is a first-order property and thus can be checked effectively as well.
    To see that non-bisimulation is recursively enumerable,
    let $\struct_v$ denote the tree-structured WTS induced by unfolding $\struct$ from configuration $v$,
    and let $\struct_v^d$ denote the finite WTS induced by restricting $\struct_v$ to up to $d$ steps from $v$.
    By Proposition~\ref{thm:PML-finite-WTS},
    two configurations $v,w$ of $\struct$ are non-bisimilar
    if and only if there is some PML formula $\phi$ such that
    $\struct_v \models \phi$ and $\struct_w \not\models \phi$.
    Since $\struct$ is bounded branching, $\phi$ can be checked by examining a finite number of configurations.
    Thus, there exists $d < \infty$ such that
    $v,w$ are non-bisimilar if and only if $\struct_v^d \models \phi$ and $\struct_w^d \not\models \phi$.
    It follows that we can locate a non-bisimilar pair by enumerating and checking all pairs $(v,w) \in E$,
    distances $d \in \Nat$, and PML formulas $\phi$ over actions $A$ and weights $P$.
    This procedure is effective by Proposition~\ref{thm:PML-finite-WTS}, which concludes the proof.
\end{proof}
\OMIT{
\subsection{Generalisation for WTSs with regular set of actions}
[TODO:
Explain that the proof rules
introduced in the previous sections can be
modified to verify bisimulation
over a generalized version of WTSs
$\langle S, \{\delta \}_{a\in\ACT}\rangle$,
where each action $a$ is a finite word and
$\ACT$ is a regular language.
Hence our main result
Theorem \ref{th:synthesis} also holds
for the generalized WTSs.]
}

\section{Learning-based bisimulation synthesis}
\label{sec:learning}
While Theorem \ref{thm:synthesis} provides a brute-force method for discovering a regular bisimulation relation, more efficient strategies are needed to identify nontrivial bisimulations in practice.
To address this challenge, we propose a learning-based approach for computing bisimulations on weakly finite regular WTSs. A transition system is \emph{weakly finite} \citep{esparza2012proving} if each configuration can reach only a finite number of configurations.
\NEW{Notably, the WTS underlying a parameterized system is weakly finite, as every configuration belongs to a finite instance of the system and thus has access to only finitely many configurations.
This local finiteness allows us to effectively compute bisimulations even though the entire system has infinitely many states.}

\NEW{Specifically, our algorithm receives as input a weakly finite regular WTS $\struct \coloneqq \WTS$, along with a regular set $I \coloneqq \{s_n\}_{n \in N} \subseteq S$ of initial configurations. The WTS $\struct$ is a regular presentation of a parameterized system $\{P_n\}_{n \in N}$ such that each $P_n$ starts in the configuration $s_n \in I$.
Given a regular set $E \subseteq S \times S$ of the bisimilar pairs to verify,
our algorithm employs active automata learning to synthesize a regular bisimulation relation $R$ such that $R \supseteq E$.}
\NEW{The learning process computes such a relation by systematically exploring bisimilar and non-bisimilar pairs within the system. A crucial requirement for this process is the ability to determine whether two configurations are bisimilar. For a parameterized system, this task amounts to computing bisimulations over a finite system instance, enabling the use of well-established verification tools from the literature \cite{garavel2022equivalence}.}
\subsection{Active automata learning}
\label{section:learning}
Automata learning \cite{angluin1987learning,rivest:inference1993,kearns:introduction1994} attempts to infer a DFA for a regular language whose definition is not directly accessible.
In \emph{active learning}, this inference is achieved by making queries to a ``teacher'' who has knowledge of the target regular language, say $\langL$.
This teacher can respond to two types of queries. The first is membership query $\MEM(w)$, asking whether a word $w$ belongs to $L$. The second is equivalence query $\EQ(\cA)$, asking whether the language $L(\cA)$ recognized by a DFA $\cA$ is identical to $L$.
An active learning algorithm performs the inference iteratively. In each iteration, it makes membership queries to gather information about $\langL$. Based on the answers, it builds a hypothesis DFA $\cA_\mathsf{hyp}$ and checks whether $\cA_\mathsf{hyp}$ is a solution through an equivalence query. If it is a solution, the algorithm terminates. If not, the teacher provides a word $u$ as a counterexample, which the algorithm will utilize to refine its hypothesis for the next iteration.

Below, we elaborate on an active learning algorithm proposed by Rivest and Schapire \cite{rivest:inference1993}, which is an improved version of \emph{L-star} \cite{angluin1987learning}. The algorithm's foundation builds upon a seminal theorem from Myhill and Nerode \cite{nerode1958linear}.
\begin{proposition} [\hspace{1sp}\cite{nerode1958linear}]
    \label{prop:myhill-nerode}
    Given a regular language $\langL$, define a relation $\equiv_\langL\,\subseteq \Sigma^* \times \Sigma^*$ such that $x \equiv_\langL y$ if and only if $\forall z\in \Sigma^*.~xz\in \langL \Leftrightarrow yz\in \langL$. Then the following are true:
    \begin{itemize}
        \item The relation $\equiv_\langL$ defines an equivalence relation. The number of distinct equivalence classes produced by this relation corresponds exactly to the number of states in the minimal DFA that can recognize the language $\langL$.
        \item Any minimal DFA recognizing $\langL$ is isomorphic to the following DFA: (i) each equivalence class $[x]$ is a state; (ii) the starting state is $[\varepsilon]$; (iii) state transitions are $[x] \to [xa]$ for $a \in \Sigma$; (iv) the accepting states are $[x]$ for $x \in \langL$.
    \end{itemize}
\end{proposition}
Specifically, in the minimal DFA recognizing $\langL$, two words $x$ and $y$ are associated with the same state if and only if no suffix $z$ can differentiate between them. In other words, $x$ and $y$ belong to different states in the minimal DFA if and only if there exists a suffix $z'$ such that $xz' \in \langL$ and $yz' \notin \langL$.
\begin{algorithm}[ht]
	\KwIn{A teacher that answers $\MEM(w)$ and $\EQ(\cA)$ about a target regular language $L$}
        \KwOut{A minimal DFA recognizing $L$}
        \smallskip
        Initialize the observation table $(W,D,T)$\;
	\Repeat{$\EQ(\cA_\mathsf{hyp}) = true$}{
	\While{$(W,D,T)$ is not closed}{
		Find a pair $(x,a)\in W\times \Sigma$ such that $\forall y\in W: \mathit{row}_D(xa)\neq \mathit{row}_D(y)$.
		Extend $W$ to $W\cup \{xa\}$ and update $T$ using membership queries accordingly\;
	}
	Build a candidate DFA $\cA_\mathsf{hyp}= (W, \Sigma, \delta, \lambda, F)$, where $\delta=\{(s,a,s') : s,s'\in W \wedge \mathit{row}_D(sa)=\mathit{row}_D(s')\}$, the empty string $\lambda$ is the initial state, and $F=\{s : T(s)=\top \wedge s\in W \}$\;
	\lIf{$\EQ(\cA_\mathsf{hyp})=(false, w)$, where $w \in L(\cA_\mathsf{hyp}) \ominus \langL$}{
		Analyse $w$ and add a suffix of $w$ to $D$}
	}
	\KwRet $\cA_\mathsf{hyp}$ as the minimal DFA for $L$\;
	\caption{Active automata learning}\label{alg:lstar}
\end{algorithm}

Algorithm~\ref{alg:lstar} presents Rivest and Schapire's version of L-star (see also \cite{chen2017learning}).
It maintains the equivalence classes induced by $\equiv_\langL$ using a {\em observation table} $(W, D, T)$, where
$W$ is a set of words representing the identified states of the minimal DFA,
$D$ is a set of suffix words distinguishing the states of the minimal DFA,
and $T$ is a mapping from $(W\cup (W\cdot \Sigma))\cdot D$ to $\{\top,\bot\}$, such that $T(w)=\top$ iff $w\in \langL$.
We write $\mathit{row}_D(x) = \mathit{row}_D(y)$ to indicate $\forall z\in D.~T(xz)=T(yz)$, meaning that states associated with the words $x$ and $y$ cannot be distinguished using only words in the set $D$ as suffix words. Observe that $x \equiv_\langL y$ implies $\mathit{row}_D(x) = \mathit{row}_D(y)$ for all $D\subseteq \Sigma^*$.
We say that an observation table is {\em closed} iff it holds that
$\forall x\in W.\,\forall a \in \Sigma.\,\exists y\in W.~\mathit{row}_D(xa)=\mathit{row}_D(y)$.
Intuitively, with a closed table, every state can find its successors with respect to all symbols in $\Sigma$. Initially, $W=D=\{\lambda\}$, and $T(w)=\MEM(w)$ for all $w\in \{\lambda\}\cup \Sigma$.

By construction, two words $x,y$ satisfying $x\equiv_\langL y$ can never be simultaneously contained in the set $W$.
When the equivalence query $\EQ(\cA)$ is $\false$, the teacher provides a counterexample $w\in L(\cA_\mathsf{hyp}) \ominus \langL$, namely the symmetric difference between $L(\cA_\mathsf{hyp})$ and $\langL$. The algorithm then performs a binary search over $w$ to find a suffix $e$ of $w$ such that $\mathit{row}_D(xa)=\mathit{row}_D(y)$ and $\mathit{row}_{D\cup \{e\}}(xa) \neq \mathit{row}_{D\cup \{e\}}(y)$ hold for some $x,y \in W$ and $a \in \Sigma$. By extending $D$ to $D\cup \{e\}$, the algorithm can identify at least one more state needed to recognize the target language $\langL$.
\begin{proposition}[\hspace{1sp}\cite{rivest:inference1993}]
\label{thm:lstar}
Algorithm~\ref{alg:lstar} can identify a minimal DFA $\cA$ that recognizes the language $\langL$. This process requires no more than $n$ equivalence queries and $n^2 + n\cdot|\Sigma| + n\log k$ membership queries, where $n$ represents the number of states in $\cA$ and $k$ is the length of the longest counterexample provided by the teacher.
\end{proposition}
To see why this proposition is true, note that the algorithm's behavior is governed by two key principles. First, each negative response to an equivalence query extends the learned DFA by at least one state. Since the Myhill-Nerode theorem constrains the candidate DFA's size to $n$ states, the number of equivalence queries is at most $n$. Second, to fully populate the observation table, the algorithm requires at most $n(n+n|\Sigma|)$ membership queries. Since the teacher provides at most $n$ counterexamples, and the algorithm employs binary search to analyze counterexamples of lengths at most $k$, the number of membership queries is bounded by $n^2 + n\!\cdot\!|\Sigma| + n\log k$.

\subsection{Learning regular bisimulations}
We now explain how to employ active automata learning to compute regular bisimulations for weakly finite regular WTSs. To this end,
we first describe a learning procedure under the so-called length-preserving assumption.
\begin{definition}[Length-preserving WTS]
    We say that a WTS $\struct \coloneqq \WTS$ is defined over alphabet $\Sigma$
    if $S$ is a subset of $\Sigma^*$. A WTS $\struct$ defined over $\Sigma$ is \emph{length-preserving} if $\delta(s,a,t) > 0$ only when $\size{s} = \size{t}$.
\end{definition}
Fix a length-preserving regular WTS $\struct \coloneqq \WTS$ over $\Sigma$, and a regular relation $E \subseteq S \times S$. We describe how to learn a regular bisimulation $R \supseteq E$ using the L-star algorithm.
Since L-star requires the target language to be unique, we aim to infer the \emph{greatest} bisimulation, which is the union of all bisimulations. Denote the target language as $\lang{\tilde R}$, the language representation of the greatest bisimulation $\tilde R$ over $\struct$.
\NEW{By the length-preserving assumption, a configuration can only reach configurations of the same size.}
Thus, we can write $\lang{\tilde R} = \bigcup_{n\ge 1} \lang{\tilde R_n}$ such that $\tilde R_n \subseteq \Sigma^n \times \Sigma^n$ is the greatest bisimulation on $\struct$ restricted to configurations of size~$n$.
\begin{figure}[tb]
    \centering
    \scalebox{0.8}{
        \begin{tikzpicture}[punkt/.style={rectangle, rounded corners, draw=black, very thick, text centered}]
        \node[punkt, minimum height=8em,  text width=4em] (learner) {};
        \node[punkt, minimum height=12.5em, text width=11.5em, right=6em of learner] (teacher) {};
        \node[align=left] (learner_text) at ($(learner) + (0, 0)$) {Learner};
        \node[align=left] (teacher_text) at ($(teacher) + (0, 1.7cm)$) {Teacher};
        \node[align=center, right=1em of teacher] (teacher_right) {a bisimulation\\relation $R \supseteq E$\\or\\a non-bisimilar\\pair $(v,u) \in E$};
        \node[rectangle, minimum height=2em, text width=7em, rounded corners, draw=black, text centered] (teacher_mem) at ($(teacher)+(0, 2.5em)$) {Is $w\in \lang{\tilde R_{\size{w}}}$?};
        \node[align=left, rectangle, minimum height=6em, text width=10em, rounded corners, draw=black] (teacher_equ) at ($(teacher)+(0, -2.5em)$)
        {(1) Is $E \subseteq R$ ?\\ (2) Is $R$ a bisimulation?\\(3) If (1) or (2) is neg-\\\hspace{1.2em}ative, find some\\\hspace{1.2em}$w \in \lang{R} \ominus \lang{\tilde R}$};
        \draw ($(learner.east)+(0, 3em)$) edge[-{To[scale=1.5]}] node[above,pos=0.4] { $\MEM(w)$ } ($(teacher_mem.west)+(0, 0.5em)$);
        \draw ($(learner.east)+(0, 2em)$) edge[{To[scale=1.5]}-] node[below,pos=0.4] { $\mathit{yes}$\,/\,$\mathit{no}$ } ($(teacher_mem.west)+(0, -0.5em)$);
        \draw ($(learner.east)+(0,-2em)$) edge[-{To[scale=1.5]}] node[above,pos=0.43] { $\EQ(\cA)$ } ($(teacher_equ.west)+(0,+0.5em)$);
        \draw ($(learner.east)+(0,-3em)$) edge[{To[scale=1.5]}-] node[below,pos=0.43] { $false, w$ } ($(teacher_equ.west)+(0,-0.5em)$);
        \draw ($(teacher_right.west)+(1.4em,0)$) edge[{To[scale=1.5]}-] ($(teacher.east)$) ;
        \end{tikzpicture}
    }
    \caption{An overview of using automata learning to synthesize a bisimulation relation $R \supseteq E$.
        Here, $\ominus$ denotes symmetric set difference, $\tilde R$ is the greatest bisimulation relation, $\tilde R_n$ is $\tilde R$ restricted to configurations of size $n$, and $R$ denotes the relation represented by automata $\cA$.}
    \label{figure:overview}
\end{figure}
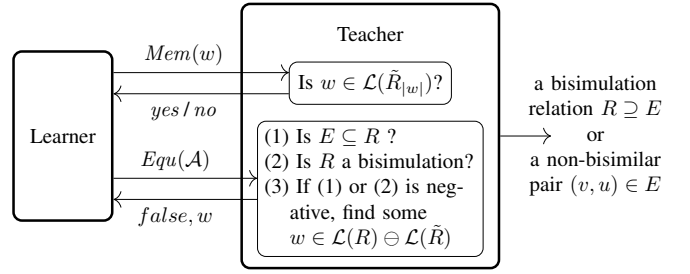

L-star needs a teacher for answering membership queries $Mem(w)$ and equivalence queries $Equ(\cA)$.
Below, we explain how to resolve these queries to learn a regular bisimulation.
\paragraph{Membership queries $Mem(w)$}
  The teacher checks whether $w = v \otimes u$ for some $v,u \in \ialphabet^*$
  such that $(v,u)$ is contained in the greatest bisimulation $\tilde R$.
  Since $\struct$ is length-preserving, this amounts to checking whether $(v,u)$ is contained in $\tilde R_{\size{w}}$.
  As $\tilde R_{\size{w}}$ is defined over a finite-state system, the teacher can answer this query by computing the fixed-length bisimulation~$\tilde R_{\size{w}}$ and checking if $(v,u)$ is in $\tilde R_{\size{w}}$.
\paragraph{Equivalence queries $Equ(\cA)$}
  The teacher checks whether $\cA$ represents a bisimulation including $E$.
  Let $R$ denote the regular relation represented by $\cA$.
  To answer this query, the teacher essentially checks that $R$ satisfies the formula~$\Phi$ at \eqref{def:Phi}, see Algorithm~\ref{alg:eqv}. It first finds a configuration pair violating the formula.
  If no such pair exists, $R$ is a bisimulation containing $E$.
  Suppose a pair $(v,u)$ of size $n$ is found.
  Then, it is a witness of either $E \not\subseteq R$ or $G_R \not\models \Phi$.
  If $(v,u) \in E \setminus \tilde R_n$, we have found a non-bisimilar pair in $E$.
  Otherwise, $(v,u)$ falls in the symmetric difference of $R$ and $\tilde R$.
  It is a valid counterexample since the learner attempts to learn the greatest bisimulation.
  The teacher thus reports $(v,u)$ for the equivalence query.
\begin{algorithm}
    \smallskip
    \KwIn{
        A length-preserving regular WTS $\struct$ over $\ialphabet$;\newline
        Regular relations $R, E \subseteq \ialphabet^* \times \ialphabet^*$;
    }
    \KwResult{
        {$(v,u)$ : a non-bisimilar pair in $E$;} \newline
        {$R$ : a bisimulation on $\struct$ such that $E \subseteq R$;}\newline
        {$(false, v \otimes u)$ : $(v,u)$ violates $R = \tilde R$;}
    }
    \smallskip
    Check whether $E \subseteq R$, and whether $\struct_R \models \Phi$ holds in the sense of Theorem~\ref{th:verify}\;
    \uIf{there exists a counterexample $(v, u)$}{
        Let $n \coloneqq \max \{ \size{v}, \size{u} \}$\;
        Compute $\tilde R_n$, the greatest bisimulation $\tilde R$ restricted to configurations of size $n$\;
        \uIf{$(v, u) \in E \setminus \tilde R_n$}{
            Output $(v ,u)$ as a non-bisimilar pair in $E$\;
            Abort the learning process\;
        }
        \Else{
            \Return{$(false, v \otimes u)$};
        }
    }
    \Else{
        Output $R$ as a bisimulation relation\;
        Abort the learning process\;
    }
    \caption{Answering equivalence queries}
    \label{alg:eqv}
\end{algorithm}

Figure~\ref{figure:overview} outlines our learning procedure for regular bisimulations. This procedure might diverge since bisimulations are not generally computable.
It is guaranteed to terminate when the greatest bisimulation $\tilde R$ is regular, though the produced bisimulation is not necessarily equal to $\tilde R$.
\begin{theorem}[Correctness]
  When the learning procedure terminates, it provides the correct answer regarding whether all configuration pairs in $E$ are bisimilar.
\end{theorem}
\begin{proof}
Note that the learning procedure terminates only when the teacher pinpoints a non-bisimilar pair $(v,u) \in E$ or a bisimulation relation $R$ such that $E \subseteq R$.
Therefore, the procedure always outputs a correct answer on termination.
\end{proof}
\begin{theorem}[Termination]\label{thm:tem}
  When the greatest bisimulation $\tilde R$ is regular, the learning procedure is guaranteed to terminate in at most $m$ iterations, where $m$ is the size of the minimal DFA recognizing $\lang{\tilde R}$.
\end{theorem}
\begin{proof}
All counterexample words reported by the teacher are contained in the symmetric difference of $\lang{\tilde R}$ and the language recognized by $\cA$. Thus, by Proposition~\ref{thm:lstar}, the learning procedure is guaranteed to terminate when $\lang{\tilde R}$ is regular. Moreover, if $\lang{\tilde R}$ can be recognized by a DFA of $m$ states, the algorithm will terminate in at most $m$ iterations.
\end{proof}
\emph{Optimization with inductive invariants.}
A natural way to optimize the learning procedure is to consider a \emph{regular} inductive invariant $\mathit{Inv}$ such that $\mathit{Inv}$ contains the set of reachable configurations. The optimization is done by simply replacing the greatest finite-length bisimulations~$\tilde R_n$ with the greatest bisimulation $\tilde R_n' \coloneqq \tilde R_n \cap (\mathit{Inv} \times \mathit{Inv})$ on the inductive invariant $Inv$ when answering the membership and equivalence query. Since $\tilde R_n'$ can be much smaller than $\tilde R_n$, replacing $\tilde R_n$ with $\tilde R_n'$ can lead to significant speed-ups. Note that a bisimulation~$R'$ on $\mathit{Inv}$ can be extended to a bisimulation~$R$ on all configurations by setting~$R \coloneqq R' \cup \{(v, v) : v \not\in \mathit{Inv}\}$, which is regular when $R'$ is regular. The inductive invariant~$\mathit{Inv}$ may be specified manually or generated automatically \cite{chen2017learning}.
\subsection{Padding WTS for bisimulation learning}
We proceed to show that it is possible to relax the length-preserving assumption in the previous section using \emph{padding}. Specifically, if the maximal size of reachable configurations is known in initial configurations, then we can reduce bisimulation inference of a weakly finite WTS $\struct$ to that of a length-preserving WTS, which is essentially a padded version of $\struct$.
\begin{definition}[Padded WTS]
    Let $\struct \coloneqq \WTS$ be a WTS over alphabet $\Sigma$.
    For a word $s \in \Sigma^*$, define a language
    $Pad(s) \coloneqq s \cdot \blank^*$.
    Namely, $Pad(s) \subseteq \Sigma_\blank$ is the set of words obtained by appending to $s$ arbitrarily many blank symbols $\blank$.
    We define the \emph{padded variant} of $\struct$ as the length-preserving WTS
    $\struct_{\mathsf{pad}} \coloneqq \langle \tilde S, P, A, \tilde \delta, +\rangle$ over alphabet $\Sigma_\blank$,
    where $\tilde S \coloneqq \bigcup_{s \in S} Pad(s)$ and
    \[
        \tilde \delta(s,a,t) \coloneqq \begin{cases}
        \delta(s',a,t'), & s \in Pad(s'),t \in Pad(t'),\size{s} = \size{t}\,;\\
        0, & otherwise.
        \end{cases}
    \]
\end{definition}
Intuitively, $\tilde S$ is obtained by padding configurations in $S$, and $\tilde \delta_a$ is the length-preserving restriction of $\delta_a$ on these padded configurations. Note that $\struct_{\mathsf{pad}}$ is definable in $\UNIV$ given a regular presentation of $\struct$. Hence, $\struct_{\mathsf{pad}}$ is effectively regular for a regular WTS $\struct$.
The following result shows that we can encode bisimilar configurations in $\struct$ such that the encoded configurations are bisimilar in $\struct_{\mathsf{pad}}$ and vice versa.
\begin{proposition}
\label{prop:padding}
    Let $\struct$ be a weakly finite WTS over alphabet $\Sigma$, and $\struct_{\mathsf{pad}}$ be the padded variant of $\struct$ over alphabet $\Sigma_\blank$.
    Then a relation $E \subseteq \Sigma^* \times \Sigma^*$ consists of bisimilar pairs over $\struct$
    if and only if $\tilde E$ consists of bisimilar pairs over $\struct_{\mathsf{pad}}$, where
    \[
        \tilde E \coloneqq \bigcup_{(v, u)\,\in\,E} \left( (Pad(v) \times Pad(u)) \cap (\Sigma_\blank \times \Sigma_\blank)^{f(v,u)}\right)
    \]
    and $f(v,u) \coloneqq {\max\{n(v), n(u)\}}$ with $n(s) \coloneqq \max \{\, \size{s'} : s'$ is reachable from $s$ in $\struct\,\}$.
\end{proposition}
\begin{proof}
    Consider a pair of configurations $(v,u) \in E$ and let $n \coloneqq {\max\{n(v), n(u)\}}$. Then $n < \infty$ since $\struct$ is weakly finite. Consider the pair of padded configurations $(v', u') \in \tilde E \cap (Pad(v) \times Pad(u))$ by padding $v$ and $u$ to size $n$. Observe that every bisimilar pair of configurations reachable from $(v,u)$ in $\struct$ has its padded version in $\tilde R_n$ and vice versa.
    Therefore, $v$ and $u$ are bisimilar in $\struct$ if and only if $(v',u') \in \tilde R_n \subseteq \tilde R$. Similarly, for each $(v',u') \in \tilde E$, $v'$ and $u'$ are bisimilar in $\struct_{\mathsf{pad}}$ if and only if their unpadded counterparts $v$ and $u$ are bisimilar in $\struct$. This concludes our proof.
\end{proof}
By Proposition~\ref{prop:padding}, checking that $E$ consists of bisimilar pairs in $\struct$ amounts to
checking that its padded version $\tilde E$ consists of bisimilar pairs in $\struct_{\mathsf{pad}}$.
To translate $E$ to $\tilde E$, we need to know $n(v)$ and $n(u)$ (i.e., the largest sizes of the configurations reachable from $v$ and $u$, respectively) for each pair $(v,u) \in E$.
In practice, given a regular presentation of a parameterized system $\{ P_n \}_{n \in N}$, we can often compute $\tilde E$ from $E$ by encoding the parameter $n \in N$ using padding \cite{lin2016liveness,chen2017learning,markgraf2020parameterized}. All the examples discussed in our case study can be adapted to this encoding.

\section{A framework for anonymity and uniformity verification}
\label{sec:anonymity-examples}
\NEW{In this section, we introduce how probabilistic bisimulation can be employed to reason about anonymity and uniformity of Markov decision processes (MDPs) and Markov chains.}
Fix an MDP $\struct \coloneqq \WTS$.
Recall that a path of $\struct$ is a sequence $s_0 \to_{a_1} s_1 \to_{a_2} \cdots$ such that $\delta(s_{i-1}, a_{i}, s_{i}) \neq 0$ for each $i$.
We will use $\pi(\struct)$ to denote the set of \emph{finite} paths of~$\struct$,
and use $\mathcal{D}_{A}$ to denote the set of probability distributions over $A$.
An \emph{adversary} $f: \pi(\struct) \to \mathcal{D}_{A}$
resolves the nondeterministic choices of $\struct$ and
induces a WTS $\struct_f \coloneqq \langle \tilde{S}, P, \tilde{A}, \tilde{\delta}, +\rangle$,
where $\tilde{S} \coloneqq \pi(\struct)$ and $\tilde{A} \coloneqq A \uplus \{\alpha\}$ with a dummy action $\alpha$.
The transition function $\tilde{\delta} : \tilde{S} \times \tilde{A} \times \tilde{S} \to P$ is defined such that for any paths $\pi \coloneqq s_0 \to_{a_1} \cdots \to_{a_n} s_n$ and $\pi' \coloneqq s_0 \to_{a_1} \cdots \to_{a_n} s_n \to_{a} s_{n+1}$
in $\pi(\struct)$ (i.e., $\pi'$ is a path extending $\pi$ with one more transition $s_{n} \to_{a} s_{n+1}$),
it holds that
$\tilde{\delta}(\pi, a, \pi') = f(\pi)(a) \cdot \delta(s_n, a, s_{n+1})$,
where $f(\pi)(a)$ denotes the probability of the adversary selecting action $a$ given the path $\pi$.
We stipulate that $f(\pi)(a) \neq 0$ only if $\delta(s_n, a, s)\neq 0$ for some $s \in S$,
i.e., the adversary can only select those actions that lead to a successor configuration in $\struct$.
We set $\tilde{\delta}(\pi, \alpha, \pi) = 1$ for every $\pi \in \tilde{S}$ ending with a terminal configuration of $\struct$.
Intuitively, $\struct_f$ describes the behavior of $\struct$ under the adversary $f$
by successively extending the paths of $\struct$ according to the actions selected by $f$.
If a path reaches a terminal configuration, it will loop there through the dummy action $\alpha$.

$\struct_f$ is a Markov chain since $\sum_{a \in A'} \sum_{\pi' \in \tilde{S}} \tilde{\delta}(\pi,a,\pi') = 1$ for every $\pi \in \tilde{S}$.
The paths of $\struct_f$ induce a probability measure that can be formalized using standard cylinder construction~\cite{woess2009denumerable}.
\OMIT{
Briefly speaking, a finite path $\ModelRun \coloneqq s_0 \to_{a_1} \cdots \to_{a_n} s_n$ of $\struct_f$
defines a \emph{basic cylinder} $Run_{\ModelRun}$, which is the set of all infinite paths of $\struct_f$ that has $\ModelRun$ as a prefix.
We associate this cylinder with the probability
$
\Pr(Run_{\ModelRun} \mid \struct_f) \coloneqq \prod_{i=1}^{n} \tilde{\delta}(s_{i-1}, a_i, s_{i}).
$
This definition gives rise to a unique probability measure for the $\sigma$-algebra
over the set of paths starting from $s_0$.
}
\NEW{Given a finite path $\ModelRun \coloneqq s_0 \to_{a_1} \cdots \to_{a_n} s_n$,
we refer to $tr(\ModelRun) \coloneqq a_1 \cdots a_n \in A^*$ as the \emph{trace} of $\ModelRun$.
We call a set of traces $\mathcal T \subseteq A^*$ a \emph{trace event}.
Let $Run_{\ModelRun}$ be the set of infinite paths with prefix $\ModelRun$.}
The probability of a trace event~$\mathcal T$ with respect to a source state $s$ is given by
\[
    \Pr\nolimits^s(\mathcal T \mid \struct_f) \coloneqq
    \Pr(\bigcup \left\{Run_\pi : tr(\pi)\in\mathcal T,
    \mbox{$s = \ModelRun[0]$} \right\} \mid \struct_f).
\]
For simplicity, we shall write $\Pr^s(\mathcal T \mid \struct_f)$ as $\Pr^s(\tau \mid \struct_f)$
when $\mathcal T = \{\tau\}$ for a single trace $\tau$.
\subsection{Anonymity verification}\label{sec:anonymity}
We will now describe how to verify anonymity properties using bisimulation.
Fix a set $I \subset S$ of initial configurations. An MDP $\struct \coloneqq \WTS$ is \emph{anonymous to an adversary $f$} if for every initial configuration $s \in I$ and trace event $\mathcal T$, the probability $\Pr^{s}(\mathcal T \mid \struct_f)$ is solely determined by $\mathcal T$ (and thus independent of $s$).
Intuitively, this means that the adversary cannot infer any information about a specific initial configuration by experimenting with the system and observing the traces.
An adversary $f: \pi(\struct) \to \mathcal{D}_{A}$ is \emph{observational} if
$f(\pi) = f(\pi')$ whenever $tr(\pi) = tr(\pi')$.
That is, the adversary has no access to the system's internal state and determines an action solely based on the trace observed thus far.
An MDP is \emph{anonymous} if it is anonymous to any observational adversary.
The following result establishes a connection between anonymity and bisimilarity.
\begin{theorem}\label{thm:anonymity}
    Let $\struct \coloneqq \WTS$ be an MDP and $f$ be an observational adversary.
    Suppose that $R \subseteq S \times S$ is a bisimulation over $\struct$.
    Then for any $(u,v) \in R$ and trace event $\mathcal T$,
    we have $\Pr^v(\mathcal T \mid \struct_f) = \Pr^u(\mathcal T \mid \struct_f)$.
    That is, $u$ and $v$ induce the same trace distribution under the intervention of~$f$.
\end{theorem}
\begin{proof}
    Fix a trace event $\mathcal T \subseteq A^*$. We can assume w.l.o.g.~that $\mathcal T$ is prefix-free.
    Indeed, if $\tau,\tau'\in \mathcal T$ and $\tau$ is a prefix of $\tau'$, then we can remove $\tau'$
    without changing the probability of $\mathcal T$.
    Thus, we have $\Pr^s(\mathcal T \mid \struct_f) = \sum_{\,\tau \in \mathcal T} \Pr^s(\tau \mid \struct_f)$,
    and it suffices to prove this theorem for the case $\mathcal T = \{\tau\}$.
    We prove it by induction on the length of $\tau$.
    For the base case, suppose that $\tau = a\in A$.
    Since $u$ and $v$ are bisimilar, we have $\sum_{s \in S} \delta(u, a, s) = \sum_{s \in S} \delta(v, a, s)$.
    It follows that
    	$\Pr\nolimits^u(\tau \mid \struct_f) = f(u)(a)\cdot \sum_{s\,\in\,S} \delta(u, a, s)
    	= f(v)(a) \cdot \sum_{s\,\in\,S} \delta(v, a, s)
    	= \Pr\nolimits^v(\tau \mid \struct_f)$
    by the definition of observational adversary.
    For the induction step, suppose that the hypothesis holds for $\size{\tau} = n$.
    Consider a trace $\tau \coloneqq a \cdot \tau'$ with $a \in A$ and $\tau' \in A^n$.
    For each $s \in S$, define an observational adversary $f_{s,a}$ such that
    \[
        f_{s,a}(s_0 \to_{a_1} \cdots \to_{a_n} s_n) \coloneqq f(s \to_a s_0 \to_{a_1} \cdots \to_{a_n} s_n).
    \]
    Then we have
    \begingroup
    \allowdisplaybreaks
    \begin{align*}
        & \Pr\nolimits^{u}(\tau \mid \struct_f)\\
        & = f(u)(a) \sum_{s\,\in\, S} \delta(u, a, s) \cdot \Pr\nolimits^{s}(\tau' \mid \struct_{f_{u,a}}) \tag*{(def.)}\\
        & = f(u)(a) \sum_{[s]\,\in\, S/R} (\sum_{w \,\in\, [s]} \delta(u, a, w)) \cdot \Pr\nolimits^{s}(\tau' \mid \struct_{f_{u,a}}) \tag*{(hypo.)} \\
        & = f(v)(a) \sum_{[s]\,\in\, S/R} (\sum_{w \,\in\, [s]} \delta(v, a, w)) \cdot \Pr\nolimits^{s}(\tau' \mid \struct_{f_{v,a}}) \tag*{(bisim.)} \\
        & = f(v)(a) \sum_{s\,\in\, S} \delta(v, a, s) \cdot \Pr\nolimits^{s}(\tau' \mid \struct_{f_{v,a}}) \tag*{(hypo.)} \\
        & = \Pr\nolimits^{v}(\tau \mid \struct_{f}), \tag*{(def.)}
    \end{align*}
    \endgroup
    where
    $S/R$ denotes the set of equivalence classes induced by the bisimulation $R$, and
    $[s] \coloneqq \{ s' \in S : (s,s') \in R \}$ is the equivalence class represented by $s$.
    Therefore the hypothesis also holds for $\size{\tau} = n + 1$. The statement hence follows by mathematical induction.
\end{proof}
Based on Theorem~\ref{thm:anonymity}, we may verify the anonymity of an MDP $\struct \coloneqq \WTS$ as follows.
Given a set $I \subseteq S$ of initial states, we specify a reference system $\struct' \coloneqq \langle S,P,A,\delta',+\rangle$ such that the trace distribution of $\struct'$ is independent of specific initial states regardless of the intervention of any observational adversary $f$.
We then find a bisimulation relation $R$ between $\struct$ and $\struct'$
such that $R \cap (I \times I)$ coincides with the identity relation over $I$.
When such a relation $R$ is found, we can conclude that the trace distribution of $\struct_f$ is also independent of the initial states for any observational adversary $f$, thereby proving the anonymity property of $\struct$.
\subsection{Uniformity verification}
\label{sec:uniformity-examples}
Bisimilarity preserves bounded termination: if two systems $\struct$ and $\struct'$ are bisimilar, then $\struct$ terminates in $n$ steps iff $\struct'$ does. We can generalize this fact and use bisimilarity to establish uniform output distributions for probabilistic programs.
\begin{definition}[\hspace{1sp}\cite{barthe2020foundations}]
    A \emph{probabilistic program} is defined by a triple $\mathcal{P} \coloneqq (\struct, I, \{F_s\}_{s\in I})$, where $\struct \coloneqq \langle S, P, \{a\}, \delta, +\rangle$ is a Markov chain with dummy action $a$, $I \subseteq S$ is a set of initial configurations, and $F_s \subseteq S$ is a set of final configurations for each $s \in I$. For simplicity, we will write $\struct$ as $\WTSS$ with probabilistic transition function $\delta: S \times S \to P$.
\end{definition}
Starting from an initial configuration $s \in I$, the probabilistic program $\mathcal{P}$ terminates when the Markov chain $\struct$ reaches some final configuration $s' \in F_s$.
The \emph{uniformity property} of $\mathcal{P}$
asserts that, from each initial configuration $s\in I$, the program has the same probability of reaching each final configuration in $F_s$ on termination. In other words, the reachability distribution over $F_s$ is uniform for each initial configuration $s \in I$.

For a Markov chain $\struct \coloneqq \WTSS$, let $\rev{\struct} \coloneqq \langle S,P,\rev{\delta} \rangle$ be a WTS where $\rev{\delta}(s,t) \coloneqq \delta(t,s)$ for all $s,t \in S$.
The following result relates uniform reachability distribution of $\struct$ to bisimilarity over $\rev{\struct}$.
\begin{theorem}
\label{thm:uniformity-proof-rule}
Let $\struct \coloneqq \WTSS$ be a Markov chain, $s_{0} \in S$ be an initial configuration, and $F \subseteq S$ be a set of final configurations.
Then $s_0$ has a uniform reachability probability over $F$ if there exists a bisimulation relation $R\subseteq S\times S$ over $\rev{\struct}$ such that (i) $F\times F \subseteq R$, and (ii) $s_0$ is bisimilar only to itself with respect to $R$.
\end{theorem}
\begin{proof}
Suppose $R$ is a bisimulation satisfying conditions (i) and (ii).
Let $S_n \coloneqq \{s \in S:$ there is a path $\pi$ from $s_0$ to $s$ such that $|\pi|=n\}$.
It is not hard to show (e.g., by induction on $n$) that if $(u,v) \in R$, then $u \in S_n \Leftrightarrow v \in S_n$ for $n\ge 0$.
Furthermore, for $s \in S$ and $n \ge 0$, define
$p_n(s) \coloneqq \sum \left\{ \Pr(Run_\pi) : \pi~\text{is a path from $s_0$ to $s$ with $|\pi|=n$}\right\}$,
i.e., $p_n(s)$ is the probability that $\struct$ reaches $s$ from $s_0$ after making precisely $n$ transitions.
Now, we claim that for each $(u,v) \in R$, it holds that $p_n(u) = p_n(v)$ for $n\ge 0$.
We prove this claim by induction on $n$.
The base case follows from condition (ii), since $n=0$ implies that $u = v = s_0$.
For the induction step, suppose that $(u,v) \in R$.
Since $u \in S_{n+1} \Leftrightarrow v \in S_{n+1}$, either both $u$ and $v$ are not in $S_{n+1}$,
or both $u$ and $v$ are in $S_{n+1}$.
In the former case, we have $p_{n+1}(u) = p_{n+1}(v) = 0$. In the latter case, we have
\begingroup
\allowdisplaybreaks
\begin{align*}
    p_{n+1}(u) & = \sum_{s\,\in\,S} p_n(s) \cdot \delta^{-1}(u,s) \tag*{(def.)}\\
        & = \sum_{[s]\,\in\,S/R} p_{n}(s) \cdot \sum_{t\,\in\,[s]} \delta^{-1}(u,t) \tag*{(hypo.)}\\
        & = \sum_{[s]\,\in\,S/R} p_{n}(s) \cdot \sum_{t\,\in\,[s]} \delta^{-1}(v,t) \tag*{(bisim.)}\\
        & = \sum_{s\,\in\,S} p_n(s) \cdot \delta^{-1}(v,s) \tag*{(hypo.)}\\
        & = p_{n+1}(v). \tag*{(def.)}
\end{align*}
\endgroup
Here, $S/R$ denotes the set of equivalence classes induced by the bisimulation $R$, and
$[s] \coloneqq \{ s' \in S : (s,s') \in R \}$ is the equivalence class represented by $s$.
By mathematical induction, we see that the hypothesis $p_n(u) = p_n(v)$ holds for all $n\ge 0$.
Finally, note that $\sum_{n\ge 0} p_n(s)$ is the reachability probability of configuration $s$. By condition (i), we have $(u,v) \in R$ for all $u,v \in F$. Thus, $p_n(u)=p_n(v)$ holds for all $u,v \in F$ and $n\ge 0$. It follows that all configurations in $F$ have the same reachability probability.
\end{proof}
\OMIT{
We say that a probabilistic program $\mathcal{P} \coloneqq (\struct, I, \{F_s\}_{s\in I})$ is \emph{regular} if $\struct$ and $I$ are regular, and
there is a regular relation $Z \subseteq S \times S$ such that $(s,s') \in Z$ if and only if $s' \in F_s$.
Technically, we can assume w.l.o.g.~that, starting from each $s \in I$, $\struct$ reaches a configuration $s'\in F \coloneqq \bigcup_{s\in I} F_s$ only if $s' \in F_s$. This assumption can be fulfilled by augmenting the Markov chain $\struct$ as follows. First, we memorize the initial configuration on each path. Thus, a path $s_0 \to s_1 \to s_2 \to \cdots$ of $\struct$ will correspond to the path $(s_0,s_0) \to (s_0,s_1) \to (s_0,s_2) \to \cdots$ of the augmented Markov chain. Second, we specify $\{(s,s) : s \in I\}$ and $Z$ as the sets of initial and final configurations, respectively. The resulting probabilistic program is regular as long as the original program $\mathcal{P}$ is regular.
By Theorem~\ref{thm:uniformity-proof-rule}, checking uniformity of c program amounts to finding a bisimulation relation $R$ over $\rev{\struct}$ such that (i) $F\times F \subseteq R$, where $F \coloneqq \bigcup_{s\in I} F_s$, and (ii) $R \cap I = \{ (s,s) : s \in I \}$.
When $\mathcal{P}$ is regular, it is decidable to check whether a candidate regular bisimulation $R$ satisfies the verification conditions (i) and (ii).
}
We say that a probabilistic program $\mathcal{P} \coloneqq (\struct, I, \{F_s\}_{s\in I})$ is \emph{regular} if $\struct$ and $I$ are regular, and
the relation $E \coloneqq \{ (v,u) \in S \times S : v,u \in F_s$ for some $s \in I \}$ is regular
(which holds when, for example, there is a regular relation $H\subseteq S \times S \times S$ such that $(v,u) \in F_s$ iff $(s,v,u) \in H$).
By Theorem~\ref{thm:uniformity-proof-rule}, checking uniformity of $\mathcal{P}$ amounts to finding a bisimulation relation $R$ over $\rev{\struct}$ satisfying $E \subseteq R$ and $R \cap (I \times I) = \{ (s,s) : s \in I \}$.
Since $\struct^{-1}$ is effectively regular when $\struct$ is regular,
for a regular probabilistic program $\mathcal{P}$, we can use Theorem~\ref{thm:synthesis} to check whether a regular relation $R$ is a proof for the uniformity of $\mathcal{P}$.

\NEW{Some remarks are in order. In our analysis of uniformity thus far, we have implicitly assumed that the final configurations are reachable with probability~1.}
Indeed, uniformity holds trivially for unreachable configurations since these configurations are ``reached'' with the same zero probability. In software model checking, it is common to break down the verification of a correctness property into separate verification tasks for partial correctness and termination \cite{ahrendt2016deductive,barthe2017proving,barthe2020foundations}.
Similarly, for assertions like ``the final configurations are reached uniformly at random,'' a uniformity proof only provides a partial correctness guarantee. We need to additionally check \emph{almost-sure termination} to establish total correctness.
Almost-sure termination of probabilistic programs can be verified using a wide array of automated techniques in the literature (e.g., \cite{lin2016liveness,esparza2012proving,abate2021learning,barthe2020foundations}) and is not the focus of this work.
\subsection{Extensions}
\label{sec:extensions}
We briefly discuss how to extend our previously introduced formalism for checking probability equivalence and handling models with parametric transition probabilities.
\paragraph{Probability equivalence}
Although we have assumed that a probabilistic program can reach the relevant final configurations with probability 1, this assumption is not essential. In fact, for a probabilistic program $(\struct, I, \{F_s\}_{s\in I})$, the uniformity properties we verify are \emph{conditional}: the reachability distribution over $F_s$ is uniform whenever the program reaches $F_s$ from the initial configuration $s \in I$. We can exploit this fact to capture and verify a property closely related to uniformity, called \emph{probability equivalence}.

More precisely, fix an initial configuration $s_0 \in I$ and an index set $J$. Given a set $S_j$ of final configurations for some $j \in J$, let $\mathcal P_j$ be the set of paths $\pi$ from $s_0$ to $S_j$. Define $\mathcal P \coloneqq \bigcup_{j \in J} \mathcal P_j$. Let $X : \mathcal P \to J$ be a random variable such that $X(\pi) = j$ indicates $\pi \in \mathcal P_j$.
Then the events $\{\mathcal P_j\}_{j\in J}$ have the same probability if and only if the conditional probability $\Pr(X \mid \mathcal P)$ yields a uniform distribution over $J$ on program termination, which can be directly checked within our framework by augmenting $\struct$. We will demonstrate how this extension applies through probability equivalence checks in random walks, random sums, and the ballot theorem.
\paragraph{Parametric probabilities}
Fix $\struct \coloneqq \WTS$. For $\struct$ to be regular, the addition operator $+$ must be regular-presentable.
It suffices to define the axioms of addition, such as associativity, commutativity, and identity element, in the first-order theory of $\struct$.
Based on this observation, we may generalize our proof rule of bisimulation to support (universal) parametric probabilities as follows.
We first extend $\struct$ with a regular set $Q$ of symbols to represent parametric probabilities, along with a new addition operator $+$ defined over $Q \cup P$.
We then encode the axioms of addition in $\FO{\struct}$ such that two probability masses are equal iff they are equal for \emph{any} feasible instantiation of the parametric probabilities.
Such encoding is possible as $\struct$ is bounded branching, e.g., when the branching is bounded by $2$, we can define the commutativity axiom by
$$
\forall x_1.\forall x_2.(x_1 \in (Q\cup P) \wedge x_2 \in (Q\cup P)) \Rightarrow (x_1+x_2=x_2+x_1).
$$
In this case, given a bisimulation $R$, if ${u \xrightarrow{1+q}_a E}$ holds for an equivalence class $E \in S/R$ and a parametric probability $q \in Q$, then $(u,v) \in R$ implies that $v \xrightarrow{1}_a t$ and $v \xrightarrow{q}_a t'$ hold for some $t,t' \in E$.
We illustrate this extension in our evaluation by proving the anonymity of the crowds protocol.
\section{Case studies}\label{sec:case-study}
\newsavebox{\mybox}
\begin{lrbox}{\mybox}
    \begin{minipage}{0.27\textwidth}
        \centering
        \begin{lstlisting}[mathescape=true]
def DCP(bool b1, b2, bv[N])
  bv[0] = bv[0] $\oplus$ b1 $\oplus$ b2
  bv[1] = bv[1] $\oplus$ b1
  bv[N-1] = bv[N-1] $\oplus$ b2
  for i in 2 ... N-1
    int b = coin()
    bv[i] = bv[i] $\oplus$ b
    bv[i-1] = bv[i-1] $\oplus$ b
  print bv
        \end{lstlisting}
    \end{minipage}
    \begin{minipage}{0.23\textwidth}
        \centering
        \begin{lstlisting}[mathescape=true]
def RandWalk(int N)
  int pos = N
  while 0 < pos < 2N
    int d = coin()
    pos = pos - 2d + 1
  print pos
       \end{lstlisting}
       \begin{lstlisting}[mathescape=true]
def RandSum(int N)
  int sum = 0
  int p = 0
  for i in 1...2N
    p = coin()
    sum = sum + p
  print sum
       \end{lstlisting}
    \end{minipage}
    \newline
    \begin{minipage}{0.2\textwidth}
        \centering
        \begin{lstlisting}[mathescape=true]
def NaiveRNG(int N)
  int x = 0
  int y = 1
  while true
    x = 2x + coin()
    y = 2y
    if y $\ge$ N
      if x < N
        break
      else
        x = 0
        y = 1
  print x
        \end{lstlisting}
    \end{minipage}
    \begin{minipage}{0.23\textwidth}
        \centering
        \begin{lstlisting}[mathescape=true]
def KnuthYaoRNG(int N)
  int x = 0
  int y = 1
  while true
    x = 2x + coin()
    y = 2y
    if y $\ge$ N
      if x < N
        break
      else
        x = x - N
        y = y - N
  print x
        \end{lstlisting}
    \end{minipage}
    \begin{minipage}{0.2\textwidth}
        \centering
        \begin{lstlisting}[mathescape=true]
def Ballot(int N)
  int a = 0
  int b = 0
  int vote = coin()
  int first = vote
  bool tied = false
  for i = 1 ... N
    if vote = 1
      a = a + 1
    else
      b = b + 1
    if a = b
      tied = true
    vote = coin()
  print (tied, first)
        \end{lstlisting}
    \end{minipage}
\end{lrbox}
\begin{figure*}[ht]
\centering
    \scalebox{0.85}{\usebox{\mybox}}
    \caption{Example probabilistic programs for verifying probability uniformity and equivalence.
    We use $\oplus$ to denote the XOR Boolean operator, and use \textsf{N} to represent an unbounded natural number parameter.
    The function \textsf{coin()} returns 0 or 1 uniformly at random.
    }\label{fig:uniformity}
\end{figure*}

\paragraph{Dining cryptographers protocol}
As we introduced in Section \ref{sec:intro}, this protocol anonymously computes the parity of the participants' secret bits.
We model it as an MDP, allowing the adversary to choose the random bits used by the observing participant $k$.
A configuration is encoded as $(s,w)$ such that $s$ is a control state, $w \in \{0, 1\}^n$ is a bit-vector, and $n$ is the participant number.
At an initial configuration, each $w[i]$ represents the secret $x_{k+i}$ held by participant $k+i$.
Thus, the interpretation of $w$ depends on who is observing.
The system has three types of transitions:
the observer choosing head or tail (via actions $\mathsf{H}$ or $\mathsf{T}$);
a non-observer tossing head or tail with probability 0.5 (both via action $\mathsf{X}$);
a participant announcing zero or one (via actions $\zero$ or $\one$).
The random bits computed by the observer are visible as actions $\mathsf{H}$ and $\mathsf{T}$,
while the random bits computed by the other participants are hidden as the dummy action $\mathsf{X}$.

Starting from an initial configuration $(s_I,w)$ with $|w| = n$, a maximal trace consists of $n$ update followed by $n$ announcements. Let $b_i$ denote the random bit computed by participant $k+i$.
For $i \in \NatInt{0}{n-1}$, the $i$-th update changes the value of $w[j]$ to $w[j] \oplus b_{i}$ for $j \in \{i, {i+1}\}$.
A configuration $(s',w')$ reached from $(s,w)$ after $n$ updates would satisfy $w'[i] = x_{k+i} \oplus b_{i} \oplus b_{i-1}$ for $i \in \NatInt{0}{n-1}$.
The trace then ``prints out'' $w'$ by going through $n$ announcements via actions
$a_0, \dots, a_{n-1}$, where $a_i$ is $\one$ if $w'[i] = 1$ and $a_i$ is $\zero$ if $w'[i] = 0$.
Here, $a_i$ is interpreted as the announcement made by participant $k+i$.
To prove anonymity, we define a reference system
where the announcements in a maximal trace starting from an initial configuration $(s,w)$
are uniformly distributed over
$\{ (a_0, \dots, a_{n-1}) : f(\overline{a}) = f(w),
~a_0 = w[0] \oplus b_0 \oplus b_{n-1}\}$.
In this way, the distribution of the announcements is independent of
the initial configuration once the values of
$x_k$, $b_0$, $b_{n-1}$, and $f(\overline{a}) \coloneqq a_0 \oplus \cdots \oplus a_{n-1}$
(i.e., the information observed by participant $k$) are fixed.
We then compute a bisimulation between the original system and the reference system and establish the desired anonymity.

In our evaluation, we also examine a generalized version of dining cryptographers where the secret messages $x_0,
\dots, x_{n-1}$ are bit-vectors of a parameterized size $m \ge 1$.
Unfortunately, the set $I$ of initial configurations is not regular in this setting.
To construct a regular model, we allow a configuration to encode secret messages of different sizes. We then devise the transition system such that an initial configuration $(s,w)$ can properly complete the protocol
(i.e., it yields a trace containing the $n$ announcements $a_0, \ldots, a_{n-1}$)
if and only if the messages encoded in $w$ have the same size.
The resulting MDP $\struct'$ overapproximates the MDP $\struct$ of the generalized protocol, as the traces of the former subsume those of the latter. Thus, we can verify $\struct'$ to establish the anonymity of $\struct$.
\paragraph{Crowds protocol}
The crowds protocol \cite{reiter1998crowds} enhances anonymous data transmission by randomly routing a message within a group of users before it reaches the true destination. When a user wants to send a message to another user, she transmits it to a random member of the group. If the receiver is corrupt, he will disclose the sender's identity. Otherwise, he will forward the message to the final destination with probability $p$, and to another random member with probability $1-p$, extending the route for one more step.

This protocol offers various levels of anonymity guarantees, many of which rely on the following fact: all potential senders are equally likely to be exposed by a corrupt user after the first forwarding.
We verify this fact as an anonymity property of an MDP.
Each configuration is encoded as $(s,x,y,n)$, where $s$ is a control state, $n$ is the crowd size, $x \in \NatInt{1}{n}$ is the user currently holding the message, and $y \in \NatInt{1}{n}$ is the original sender. The initial configurations are of the form $(s_I,i,i,n)$ for $1 \le i \le n$. The routing process operates in rounds. In each round, the MDP randomly selects a user $z \in \NatInt{1}{n}$ to receive the message from $x$.
Then, an adversary determines whether the receiver $z$ is corrupt. If the receiver is corrupt, the MDP emits action $\mathsf{T}$ when $x=y$, indicating that the original sender is revealed, and emits action $\mathsf{F}$ when $x\neq y$, indicating that a false sender is revealed. If the receiver is not corrupt, he forwards the message to the final destination with probability $p$, and proceeds to the next round with probability $1-p$ and a dummy action $\mathsf{X}$. A trace extends indefinitely until the message is sent to a corrupt user or to the final destination.
We verify that all initial configurations in the same crowd are bisimilar, meaning that all potential senders are equally likely to be exposed by an external observer.
\paragraph{Grades protocol}
The grades protocol \citecf{apex,Kiefer2013} is a multi-party computation algorithm aiming to securely compute the sum of the secrets held by the participants.
The setting of the protocol is pretty similar to that of the dining cryptographers:
given two parameters $n\ge 3$ and $g \ge 2$,
we have $n$ participants where each participant $i$ holds an \emph{integer} secret $x_i \in \NatInt{0}{g-1}$.
The goal is to compute $x_0 + \cdots + x_{n-1}$ without revealing information about the individual secrets.
Define $m \coloneqq (g-1) \cdot n + 1$.
The protocol has two steps:
(i) Each two adjacent participants $i$ and $i+1$ compute a random number $y_i \in \NatInt{0}{m-1}$;
(ii) Each participant $i$ announces $a_i \coloneqq (x_i + y_i - y_{i-1})~{\rm mod}~m$ to the other participants.
Thus, the participants can compute $h(\overline{a}) \coloneqq a_0 + \cdots + a_{n-1}~{\rm mod}~m$,
which is equal to $x_0 + \cdots + x_{n-1}$ as the $y_i$'s are canceled out due to the modulo.
The anonymity property asserts that no participant can infer the secrets held by the other participants from the observed information.
We consider a variant of grades protocol where $m = 2^k \ge (g-1) \cdot n + 1$ with parameters $k$, $n$, and $g$.
Clearly, the original protocol's anonymity and correctness properties also hold for this variant.
Since $m$ is a power of 2, the MDP of this protocol is similar to the one we constructed for the generalized dining cryptographers, except that the xor operations are replaced with additions and negations when appropriate.
A reference system is specified such that
the announcements observed by participant $k$ are uniformly distributed over
$\{ (a_0,\ldots,a_{n-1}) : h(\overline{a}) = h(\overline{x}),~a_0 = x_k + y_0 + y_{n-1}$ mod $m\}$.
We then establish the anonymity property by computing a bisimulation between
the original system and the reference system.
\paragraph{Dining cryptographers protocol (version 2)}
In this example, we verify the anonymity of the dining cryptographers by formulating it as a probabilistic program. Let $x_0,\ldots, x_{n-1}$ denote the secret bits and $b_0,\ldots,b_{n-1}$ denote the random bits computed during execution, such that the announcement made by participant $i$ is $a_i \coloneqq x_i \oplus b_{i} \oplus b_{i-1}$.
To show that an observing participant $k$ cannot infer information beyond $x_k$, $b_k$, $b_{k-1}$, and $f(\overline{x}) \coloneqq x_0 \oplus \cdots \oplus x_{n-1}$, it suffices to check that the random vector $(a_0, \dots, a_{n-1})$ is uniformly distributed over $\{ \overline{a} \in \{0,1\}^n : f(\overline{a}) = f(\overline{x}),~a_k = x_k \oplus b_k \oplus b_{k-1}\}$. Based on this, we model the protocol as the program $\mathsf{DCP}$ in Figure~\ref{fig:uniformity}. The program has arguments $b_1, b_2 \in \{0,1\}$ and $bv \in \{0,1\}^n$, whose semantics depends on who the observer is:
when participant $k$ is observing, $b_1$ and $b_2$ are interpreted to be the values of the random bits $b_{k}$ and $b_{k-1}$, respectively, while $bv[i]$ is interpreted to be the value of $x_{k+i}$ for $i \in \NatInt{0}{n-1}$,
where the indices are computed modulo $n$ as before.
The program $\mathsf{DCP}$ computes the announcements $a_{k}, \ldots, a_{k+n-1}$ in order, and stores the result in ${bv}[0],\ldots,{bv}[n-1]$.
We specify this probabilistic program as $(\struct, I, \{F_s\}_{s\in I})$, where each configuration $(s,\args{y},\args{x})$ consists of a program location $s$, the valuation $\args{y}$ of $bv$, and the secret values $\args{x}$.
We define $I = \{(s_I, \args{x}, \args{x}) : \args{x} \in \{0,1\}^n,~n \ge 3 \}$
and $F_{(s_I,\args{x},\args{x})} = \{ (s_F,\args{y},\args{x}) : \args{y} \in \{0,1\}^n,~f(\args{y}) = f(\args{x}),~y_k = x_k \oplus b_1 \oplus b_2 \}$,
where $k$ is the observer, $s_I$ is the initial program location, and $s_F$ is a final program location.
We verify the uniformity of this program, which suffices to establish the anonymity of the dining cryptographers protocol.
\paragraph{Random walk and random sum}
We examine the probability equivalence properties of two probabilistic programs $\mathsf{RandWalk}$ and $\mathsf{RandSum}$, see Figure~\ref{fig:uniformity}. In $\mathsf{RandWalk}$, the walker starts from a position $n\ge 1$ and keeps moving leftward or rightward with equal probability. We check that the walker reaches positions $0$ and $2n$ equally likely. The probabilistic program $(\struct, I, \{F_s\}_{s\in I})$ has $I \coloneqq \{ (s_I, n, n) : n \ge 1 \}$ and $F_{(s_I,n,n)} \coloneqq \{(s_F,0,n),(s_F,2n,n)\}$, with $s_I$ denoting the initial program location and $s_F$ denoting a final program location.
In $\mathsf{RandSum}$, the program computes the sum of $2n$ random bits, and we check that
$\Pr(sum = k) = \Pr(sum = 2n-k)$
holds for $k \in \NatInt{0}{n}$ on termination. The probabilistic program $(\struct, I, \{F_s\}_{s\in I})$ has $I \coloneqq \{ (s_I,0,n) : n \ge 1 \}$ and $F_{(s_I,0,n)} \coloneqq \bigcup_{0\le k \le n} E_{k,n} \coloneqq \bigcup_{0\le k \le n} \{(s_F,k,n),(s_F,2n-k,n)\}$ for $n \ge 1$.
We verify that the reachability probability on each $E_{k,n}$ is uniform.
In both programs, uniformity implies the desired probability equivalence property by construction.
\paragraph{Random number generation}
It is well-known that when $n$ is a power of 2, one can sample a random number $x$ from $\{0,\ldots,n-1\}$ using $\lg n$ random bits.
For general $n \ge 1$, we may compute $x$ by repeatedly sampling it from $\{0,\ldots,2^{\lceil \lg n \rceil}-1\}$ until $x \le n-1$,
at which point $x$'s value is uniformly distributed over $\{0,\ldots,n-1\}$.
This procedure terminates with probability 1 and uses $2\lceil \lg n \rceil$ random bits on average.
Knuth and Yao \cite{knuth1976complexity} improved the procedure, obtaining a version using $\lceil \lg n \rceil + \Theta(1)$ random bits on average. We model these two random number generators as probabilistic programs $\mathsf{NaiveRNG}$ and $\mathsf{KnuthYaoRNG}$, respectively.
Both programs have initial states $I \coloneqq \{ q_n : n \ge 1\} \coloneqq \{(s_I,0,1,n) : n \ge 1\}$ and final states $F_{q_n} \coloneqq \{(s_F,x,y,n) : 0\le x<n,~y\ge n\}$ for each $n\ge 1$.
We verify the uniformity of these programs, proving that the algorithms indeed compute a uniform random variable over $\{0,\ldots,n-1\}$ on termination.
\paragraph{Ballot theorem}
We consider a lemma for proving the ballot theorem \cite{albarghouthi2018constraint}, which concerns the vote counting process for two candidates. Let $n_A$ and $n_B$ be parameters denoting the votes received by candidates $A$ and $B$, respectively.
Suppose $n_A > n_B$ and the votes are counted in a uniformly random sequence. Then, the theorem asserts that the probability of $A$ maintaining a lead throughout the entire counting process is $(n_A-n_B)/(n_A+n_B)$.
One technique, often known as the reflection principle, to prove this theorem involves showing that the probability of counting the first vote for $A$ and subsequently reaching a tie is equal to the probability of counting the first vote for $B$ and then reaching a tie. We design a probabilistic program $(\struct, I, \{F_{s_n}\}_{s_n\in I})$
to simulate the vote counting process.
A program configuration $(s, a, b, f, t, n)$ records the program location $s$ and maintains the current votes $a\ge 0$ for $A$,
the current votes $b \ge 0$ for $B$, $t\in \{\bot,\top\}$ indicating whether a tie has occurred, $f\in\{\bot,\top\}$ indicating whether the first vote is for candidate $A$, and the total number of votes $n$.
We specify $I \coloneqq \{ q_n : n \ge 1\} = \{(s_I,0,0,\bot,\bot,n) : n \ge 1\}$ and $F_{q_n} \coloneqq \{(s_F, a, b, \top, t, n) : a+b=n \}$.
Our program simulates vote counting by drawing uniform samples for the votes \cite{albarghouthi2018constraint}.
Proving the reflection principle then amounts to showing that, starting from each initial configuration $q_n \in I$, the program reaches $F_{q_n}$ with uniform probability.

\section{Evaluation}
\label{sec:bisimulation-experiements}
To evaluate our approach, we developed a prototype tool in Scala. This tool can process parameterized systems specified in several languages, including $\UNIV$, WS1S, and Armoise \citep{Armoise}. It utilizes toolkits from \textsc{Mona} \citep{klarlund2001mona} and \textsc{TaPAS} \citep{leroux2009tapas} to generate automata representations of WTSs and verification conditions for probabilistic bisimulations.
Our tool is built around the active learning procedure described in Section~\ref{sec:learning}.
Specifically, membership queries are resolved by computing bisimulations for finite-state system instances, while equivalence queries leverage \textsc{Mona} to verify hypothesis automata and generate counterexamples when necessary.
Candidate bisimulations are restricted to binary relations over configuration pairs with valid encoding to streamline proof inference. Below, we discuss the effectiveness, performance, and limitations of our approach based on experiments conducted on a Windows laptop with a 2.3GHz Intel i7-11800H processor and 16GB memory limit.

\NEW{
\paragraph{Effectiveness}
The first two tables in Table~\ref{tab:experiment} summarize the examples in Section~\ref{sec:case-study}, presenting the sizes of synthesized proofs (states \#S and transitions \#T) and the runtimes of \textsc{Mona} verification (Mona), bisimulation computation (Bisim), learner processing (Lstar), and tool execution (Total).
These results demonstrate that our tool effectively identifies regular proofs for all examined examples.
They also highlight proof complexity and bisimulation computation as the primary performance bottlenecks. Indeed, as proof complexity grows, the search space expands, leading to longer verification times. Also, some examples compute bisimulations over large system instances before convergence, even though the final proof sizes remain relatively small (e.g., the DCP multi-bit example).}
\begin{table}[t]
\centering
\caption{Experimental Results}
\label{tab:experiment}
\begin{adjustbox}{width=.95\columnwidth}
\begin{tabular}{|l||c|c||c|c|c|c|}
    \hline
    Anonymity Examples & \#S & \#T & Mona & Bisim & Lstar & Total\tabularnewline
    \hline
    \hline
    DCP, single-bit & 13 & 832 & 1.5s & 2.9s & 0.1s & 5s \tabularnewline
    \hline
    DCP, multi-bit & 16 & 1024 & 1.6s & 18s & 0.2s & 20s \tabularnewline
    \hline
    Crowds Protocol & 20 & 1280 & 2.0s & 0.5s & 0.8s & 4s \tabularnewline
    \hline
    Grades Protocol & 25 & 1600 & 3.1s & 23s & 0.1s & 27s \tabularnewline
    \hline
\end{tabular}
\end{adjustbox}
\\\vspace{1em}
\begin{adjustbox}{width=.95\columnwidth}
    \begin{tabular}{|l||c|c||c|c|c|c|}
    \hline
    Uniformity Examples & \#S & \#T & Mona & Bisim & Lstar & Total\tabularnewline
    \hline
    \hline
    Dining Cryptographers & 4 & 256 & 0.4s & 0.3s & 0.5s & 2s \tabularnewline
    \hline
    Random Walk & 6 & 96 & 0.3s & 0.6s & 0.3s & 2s \tabularnewline
    \hline
    Random Sum & 7 & 448 & 0.8s & 1.1s & 1.2s & 3s \tabularnewline
    \hline
    Knuth-Yao RNG & 13 & 832 & 1.6s & 0.5s & 0.8s & 3s \tabularnewline
    \hline
    Naive RNG & 21 & 1344 & 2.3s & 1.4s & 0.6s & 5s \tabularnewline
    \hline
    Ballot Theorem & 52 & 3328 & 14s & 59s & 4.2s & 78s \tabularnewline
    \hline
\end{tabular}
\end{adjustbox}
\\\vspace{1em}
\centering
\begin{adjustbox}{width=.96\columnwidth}
\begin{tabular}{|l||c|c||c|c|c|c|}
    \hline
    Random Walk Settings & \#S & \#T & Mona & Bisim & Lstar & Total\tabularnewline
    \hline
    \hline
    $d=1,~k=10$ & 6 & 96 & 0.5s & 1.8s & 0.1s & 3s \tabularnewline
    \hline
    $d=1,~k=50$ & 6 & 96 & 0.5s & 53s & 0.3s & 54s \tabularnewline
    \hline
    $d=1,~k=100$ & 6 & 96 & 0.5s & 529s & 0.4s & 531s \tabularnewline
    \hline
    $d=1,~k=150$ & 6 & 96 & 0.5s & 4863s & 0.6s & 4865s \tabularnewline
    \hline
    $d=2,~k=1$ & 32 & 2048 & 8.3s & 28s & 4.5s & 42s \tabularnewline
    \hline
    $d=3,~k=1$ & 200 & 51200 & 389s & 2685s & 1167s & 4256s \tabularnewline
    \hline
\end{tabular}
\end{adjustbox}
\end{table}

\NEW{
\paragraph{Performance}
To further investigate the performance factors of the learning algorithm, we consider a random walk example within $\{ \overline{x}\in \mathbb{Z}^d : \forall i.\:|x_i|\le n \wedge k \le |x_i| \}$, where $n$ is a parameter, and $k$ and $d$ are constants.
This example can be formulated as a parameterized system $\{P_{n}\}_{n\ge k}$.
We verify that, starting from the origin, the walker reaches the $2^d$ corner points $p \in \{-n,n\}^d$ with equal probability. The third table in Table~\ref{tab:experiment} outlines our tool's performance under different $k$ and $d$.}
\NEW{
By specifying a larger $k$ in the parameterized system, we require the learner to infer proofs based on bisimulations over larger system instances, amplifying the computational burden even though the inferred proofs remain unchanged. Also, as the dimension $d$ increases, proof complexity escalates and impacts overall performance.
These results suggest that the learning efficiency largely depends on the complexity of the bisimulations and candidate solutions navigated by the learner before it arrives at the final answer.
}

\NEW{
\paragraph{Limitations}
Although our learning-based approach successfully verifies all the examined examples, it is important to note that the general problem remains undecidable.
Indeed, compared to explicit enumeration (cf.~Theorem \ref{thm:synthesis}), the learning algorithm prioritizes fast convergence over completeness, which means it may fail to find a regular proof even when one exists.
In such cases, one could consider combining our learning approach with an enumerative method like solver-based synthesis \cite{heule2013software}.}
It remains an open question to characterize a natural class of parameterized systems for which our learning algorithm is complete.
\section{Related work}
\label{sec:related-work}
Our verification framework can be construed as a parameterized variant of probabilistic model checking for anonymity and uniformity.
We discuss below the most pertinent literature.
\paragraph{Anonymity}
Formal verification of anonymous and secure communication protocols primarily relies on two methodologies: theorem proving and automated verification.
While theorem proving methods \cite{mciver2019thousand,barthe2009formal,avanzini2024quantitative} are capable of handling complex properties and systems, these methods often demand nontrivial manual effort and domain expertise when the system induces an unbounded model, which poses an obstacle to automation.
Automated verification methods offer various approaches to reason about anonymity.
Model checkers based on epistemic logic \citeeg{al2011abstraction,lomuscio2017mcmas,knapik2010parametric}
and process algebras \citeeg{cheval2014apte,tiu2010automating}
can analyze qualitative anonymity properties by abstracting randomness into nondeterminism and applying techniques like symbolic model checking \citecf{cortier2011survey}.
Probabilistic model checkers, including \textsc{Mcsta} \citep{hartmanns2014modest}, \textsc{Prism} \citep{shmatikov2004probabilistic,kwiatkowska2011prism}, and \textsc{Storm} \citep{hensel2022probabilistic}, can express quantitative anonymity properties in probabilistic temporal logic and analyze them through \NEW{exact computation or numerical approximation \cite{rutten2004mathematical}.}
Equivalence checkers \citeeg{apex,Kiefer2013,bauer2018model,cheval2022symbolic} formulate anonymity as indistinguishability of system executions, thereby reducing the verification tasks to solving language or trace equivalence problems in probabilistic systems.
Despite extensive tool support, most existing equivalence and model checkers are only capable of verifying our case studies in the \emph{finite} setting.
To the best of our knowledge, this work provides the first fully automated approach that can verify them in the parameterized setting.
\paragraph{Uniformity}
Uniformity verification is a specialized form of relational verification for probabilistic programs
\cite{barthe2020foundations}.
For finite-state programs, various tools \citeeg{kwiatkowska2011prism,hensel2022probabilistic} can be employed to model check uniformity properties.
For infinite-state programs, Barthe et al.~\cite{barthe2017proving} have extended the probabilistic program logic pRHL \citep{barthe2009formal,avanzini2024quantitative}, initially designed for relational properties, to reason about uniformity using \emph{coupling}.
Albarghouthi and Hsu \cite{albarghouthi2018constraint} further exploited program synthesis techniques to construct coupling proofs. Their work is, to the best of our knowledge, the only fully automated approach for infinite-state uniformity verification aside from our method. Notably, our tool successfully verifies all uniformity examples considered in \cite{albarghouthi2018constraint}.
A coupling argument aims to show a one-to-one correspondence between relevant execution paths, while bisimulation establishes equivalence between path probabilities. One-to-one correspondence is a stronger condition for probability equivalence by proving uniformity modulo permutation of paths. In comparison, bisimulation proves uniformity modulo summation of path probabilities.
Figure~\ref{fig:coupling} presents a toy example whose uniformity is trivial for bisimulation proofs but not directly amenable to coupling proofs \cite{albarghouthi2018constraint}.
On the other hand, coupling arguments can establish probability independence through self-composition \citep{barthe2017proving}. It remains unclear whether probability independence is provable by bisimulation when the output distribution is non-uniform.
Symbolic inference \citeeg{gehr2016psi,cusumano2018incremental,susag2022symbolic} provides automated methods to answer symbolic queries about distributions induced by probabilistic programs.
Though it is possible to encode uniformity queries in such methods, existing formalisms of symbolic inference \citecf{barthe2020foundations}
fall short in specifying parameterized systems like those considered by this work.
\begin{figure}
\centering
\scalebox{0.6}{
\begin{tikzpicture}[->, thick, >=latex, node distance=2.2cm, auto, bend angle=25, shorten >=1pt, shorten <=1pt]
  \node[state,initial]            (S)                {$s$};
  \node[state]            (T) [above right=1.2cm and 3cm of S] {};
  \node[state]            (M) [below right=.3cm and .8cm of T]   {};
  \node[state,accepting, accepting distance=6mm]  (R1) [right=2.6cm of T]    {$q_{1}$};
  \node[state,accepting]  (R2) [right=1cm of R1]    {$q_{2}$};
  \node[state]            (B) [below=.85cm of R1,xshift=-.5cm] {};
  \path[->]
    (S) edge[bend left]          node[above] {$1/3$}  (T)
    (S) edge                     node[below] {$1/3$}  (M)
    (S) edge[bend right]         node[above,pos=.7,yshift=-.05cm] {$1/3$} (B)
    (T) edge                     node[pos=0.3,xshift=-.4cm] {$1/3$} (S)
    (T) edge                     node[pos=0.5,yshift=-.1cm] {$2/3$} (R1)
    (M) edge[bend left]          node[below,xshift=.2cm,yshift=.05cm] {$2/3$} (S)
    (M) edge                     node[below,xshift=.1cm] {$1/3$} (R1)
    (B) edge                     node[yshift=-.1cm] {$1$} (R2);
\end{tikzpicture}}
\caption{A Markov chain with uniform output distribution over $F_s = \{q_1, q_2\}$}
\label{fig:coupling}
\end{figure}
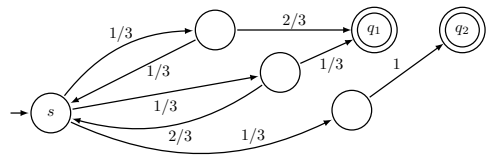
\paragraph{Bisimulation}
Our method leverages bisimulation to reason about anonymity and uniformity.
The concept of bisimulation we consider is also referred to as \emph{strong bisimulation}, and the corresponding behavioral equivalence is termed \emph{equivalence modulo strong bisimilarity}
\citecf{milner1989communication,hennessy2012exploring}.
In the literature, other types of behavioral equivalence have been studied for probabilistic systems,
including {weak bisimilarity} \citeeg{baier1997weak,philippou2000weak}
{branching bisimilarity} \citeeg{andova2006branching},
and {$\varepsilon$-bisimilarity} \citecf{spork2024spectrum},
with most investigations focusing on finite-state systems.
\NEW{For infinite-state systems, various decidable classes have been identified for strong bisimilarity, including several types of process rewrite systems and pushdown systems in both probabilistic and non-probabilistic settings \cite{srba2004roadmap,aceto2012algorithmics,forejt2018game}.
Despite these theoretical advances, few of the results have been adapted into practical verification tools \cite{garavel2022equivalence}.}
Interestingly, Forejt et al.~\cite{forejt2018game} showed that strong bisimilarity on probabilistic systems
can be reduced to strong bisimilarity on nondeterministic LTSs.
Thus, specialized proof rules are not essential for verifying probabilistic bisimulations.
Unfortunately, their reduction does not preserve regular system encoding, rendering it incompatible with our regular verification framework.
In a recent study, Abate et al.~\cite{abate2024bisimulation} proposed a data-driven approach to synthesize bisimulations for infinite-state LTSs.
Their approach utilizes SMT solvers to generate candidate bisimulations in a learner-verifier architecture
conceptually similar to our learner-teacher framework.
Nevertheless, their method is limited to learning bisimulation relations with finitely many equivalence classes,
which is often insufficient for parameterized systems.

\section{Conclusion}
This paper introduces a first-order framework for checking strong bisimulation equivalence in infinite-state probabilistic systems, with applications to anonymity and uniformity verification.
Our approach requires that (i) the examined system has an effective regular presentation, (ii) the system is bounded branching, and (iii) the system is weakly finite, which holds naturally for parameterized systems.
We show that, while the general verification problem is undecidable, our framework can effectively encode and automatically verify challenging examples from the literature.

Future research could explore generalizations to weaker versions of bisimulation equivalence like $\varepsilon$-bisimilarity \citecf{spork2024spectrum} and bisimulation metrics \citecf{van2005behavioural}, which tolerate slight deviations when comparing system behaviors. Such relaxation is particularly relevant for verifying cryptographic protocols, since most practical protocols do not achieve perfect secrecy but are still sufficiently secure. Another interesting direction is to enhance the expressiveness of our framework by utilizing the recent development of \emph{regular abstraction} \citeeg{hong2024regular,czerner2024computing}, which allows for specifying and reasoning about regular structures in background SMT theories. Finally, it is possible to improve our framework's capabilities to handle more complex systems by incorporating probabilistic model checkers and program verifiers, e.g., as oracles for bisimulation learning \cite{abate2024bisimulation}.

\linespread{.9}\selectfont
\bibliographystyle{IEEEtran}
\begingroup
\small
\bibliography{references}
\endgroup
\end{document}